\documentclass[journal]{IEEEtran}

\usepackage[noadjust]{cite}

\usepackage{mathtools} 
\usepackage{amsmath}
\usepackage{amssymb}
\usepackage{mathtools}
\usepackage{amsthm}
\usepackage{xcolor}
\usepackage{subcaption}
\usepackage{float}
\usepackage{bbm}
\usepackage[ruled,vlined]{algorithm2e}
\SetKwInput{KwInput}{Input}
\SetKwInput{KwOutput}{Output}
\SetKwComment{Comment}{/*}{ */}
\let\oldnl\nl
\newcommand{\nonl}{\renewcommand{\nl}{\let\nl\oldnl}}

\theoremstyle{plain}
\newtheorem{theorem}{Theorem}[section]

\newtheorem{lemma}[theorem]{Lemma}
\newtheorem{corollary}[theorem]{Corollary}
\theoremstyle{definition}
\newtheorem{definition}[theorem]{Definition}

\theoremstyle{remark}
\newtheorem{remark}[theorem]{Remark}
\newtheorem{example}{Example}

\usepackage[a4paper, top=1in, bottom=1in, left=0.7in, right=0.7in]{geometry}

\usepackage{graphicx}
\usepackage{epstopdf}
\usepackage{enumerate}
\usepackage{cite}
\usepackage{url} 
\usepackage{xcolor}
\usepackage{hyperref}
\usepackage{mathtools}
\usepackage{amsthm}
\usepackage{amsfonts}
\usepackage{amssymb}
\usepackage{amstext}
\usepackage{bbm}
\usepackage{enumitem}
\usepackage{hyperref}
\usepackage{float}
\usepackage{mathrsfs}

\newcommand{\Prob}{\mathsf{P}}
\newcommand{\Expect}{\mathsf{E}}
\DeclareMathOperator*{\esssup}{ess\,sup}

\newcommand{\mf}{\mathcal{F}}
\newcommand{\mg}{\mathcal{G}}
\newcommand{\mb}{\mathcal{B}}
\newcommand{\mpp}{\mathcal{P}}

\newcommand{\fb}{F, \ta}
\newcommand{\gb}{G, \ta}
\newcommand{\ta}{\theta}

\onecolumn

\begin{document}

\title{Robust Quickest Change Detection in Multi-Stream Non-Stationary Processes}
\author{Yingze~Hou,
        Hoda~Bidkhori,
        and~Taposh~Banerjee
\thanks{Y. Hou and T. Banerjee are with the Department of Industrial Engineering, University of Pittsburgh, Pittsburgh,
PA, 15260 USA e-mail: taposh.banerjee@pitt.edu.}
\thanks{H. Bidkhori is with the Department of Computational and Data Sciences, 
George Mason University, Fairfax, 
VA 22030 USA.}
}

\maketitle

\begin{abstract}
The problem of robust quickest change detection (QCD) in non-stationary processes under a multi-stream setting is studied.  
In classical QCD theory, optimal solutions are developed to detect a sudden change in the distribution of stationary data. Most studies have focused on single-stream data.
In non-stationary processes, the data distribution both before and after change varies with time and is not precisely known. The multi-dimension data even complicates such issues.  
It is shown that if the non-stationary family for each dimension or stream has a least favorable law (LFL) or distribution in a well-defined sense, then the algorithm designed using the LFLs is robust optimal. The notion of LFL defined in this work differs from the classical definitions due to the dependence of the post-change model on the change point.  
Examples of multi-stream non-stationary processes encountered in public health monitoring and aviation applications are provided.
Our robust algorithm is applied to simulated and real data to show its effectiveness.
\end{abstract}


\IEEEpeerreviewmaketitle

\section{Introduction}
\label{sec:intro}

In the classical Quickest Change Detection (QCD) theory, the objective is to detect a change in the distribution of a stochastic process with minimal delay, while controlling the frequency of false alarms (\cite{veer-bane-elsevierbook-2013, tart-niki-bass-2014, poor-hadj-qcd-book-2009}). In the independent and identically distributed (i.i.d.) setting, a sequence of random variables ${X_n}$ is observed. Before a change point, denoted as $\nu$, these variables follow a fixed distribution with density $f$. After $\nu$, they follow a different distribution with density $g$. Optimal or asymptotically optimal methods are developed to detect this distributional shift from $f$ to $g$. In the Bayesian framework, the change point $\nu$ is treated as a random variable with a known prior distribution. The performance criterion based on the average detection delay leads to the Shiryaev test as the optimal solution (\cite{shir-siamtpa-1963, tart-veer-siamtpa-2005}). Recent studies continue to advance QCD in Bayesian settings (\cite{veer-bane-elsevierbook-2013, tart-niki-bass-2014, tart-book-2019, bane-tit-2021, guo2023bayesian, naha2024bayesian}).

In non-Bayesian settings (\cite{poll-astat-1985, lord-amstat-1971}), the change point $\nu$ is considered an unknown constant. In this case, the notion of conditional delay (conditioned on the change point) can be defined, and the average delay in detection depends on the specific location of the change point $\nu$. As a result, a minimax approach is typically adopted, seeking to minimize the worst-case delay across all possible change points.
A precise mathematical formulation of this approach is provided in Section~\ref{sec:QCD model}. For the i.i.d. case, the optimal solution is the Cumulative Sum (CUSUM) algorithm, as developed in seminal works (\cite{page-biometrica-1954, lord-amstat-1971, mous-astat-1986, lai-ieeetit-1998}). A comprehensive literature review and recent developments in the non-Bayesian setting can be found in 
\cite{tart-niki-bass-2014, veer-bane-elsevierbook-2013, tart-book-2019, liang2022quickest, brucks2023modeling}. 

The problem of QCD is also studied in a multi-stream setting, where the data being processed consists of multiple independent streams of observations. At the change point, the distribution of some unknown subsets of streams can change.
It has been shown that the optimal algorithms in this setting are extensions of the optimal algorithms for the single-stream case, modified using a generalized likelihood ratio or mixture approach (\cite{tart-niki-bass-2014, veer-bane-elsevierbook-2013, tart-book-2019}). For some recent results on multi-stream QCD, see \cite{xie2013sequential, wang2015large, fellouris2017multistream, xu2021optimum, xu2021multi, xu2022active, oleyaeimotlagh2023quickest}. 

In many applications of multi-stream QCD, the multi-stream data encountered is non-stationary. This means that the distribution of the data both before and after the change is time-varying. Also, due to the non-stationary nature of this data, the distributions are unknown and also impossible to learn from limited data. We give below two examples of such non-stationary data: 
\begin{enumerate}
    \item [(1)] Public health applications:
    In the post-COVID era, rapid detection of emerging pandemics is critical (\cite{liang2022quickest, hou2024robust}). 
    In Figure \ref{fig: covid infection}, we plot the daily infection counts for Alabama and Pennsylvania during the first 150 days starting from January 22, 2020. As illustrated, infection rates in different counties evolve over time, with some experiencing early surges that later stabilize, while others exhibit slower initial growth followed by rapid increases later on;

    \item [(2)] Aviation applications:
    A classical problem in aviation is detecting approaching aircraft or objects (\cite{brucks2023modeling}). The arrival of aircraft typically leads to a stochastically growing process. For instance, Figure \ref{fig: flight distance} displays the distance and signal measurements from aircraft trajectory datasets collected around Pittsburgh-Butler Regional Airport (\cite{Patrikar2021}). The trajectories show variations in the approach patterns, highlighting the need for accurate detection algorithms.
\end{enumerate}

\begin{figure}
	\centering
	\includegraphics[scale=0.5]{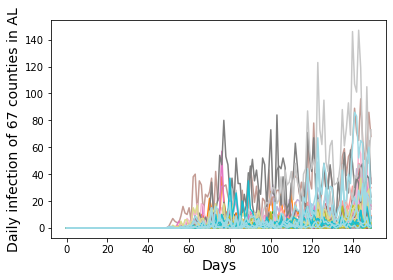}
	\includegraphics[scale=0.5]{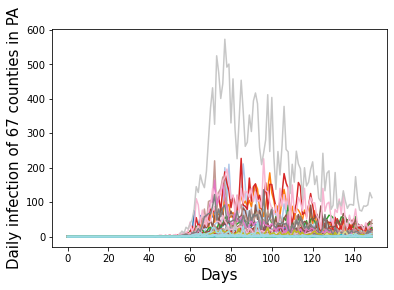}
    \caption{Daily infection rates for each county in Alabama (AL) and Pennsylvania (PA) state in the first $150$ days starting January 22, 2020. The number of infections over time (even beyond the dates shown here) has multiple cycles of high values and low values. Some counties have confirmed cases earlier than others.}
    \label{fig: covid infection}
\end{figure}

\begin{figure}
	\centering
	\includegraphics[scale=0.5]{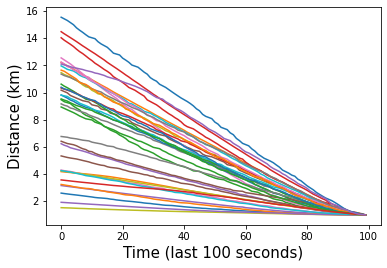}
	\includegraphics[scale=0.5]{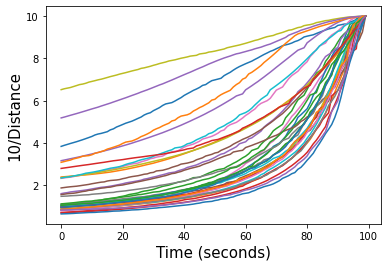}
    \caption{Distance measurements and corresponding signals extracted from a database on aircraft trajectories collected from aircraft around Pittsburgh-Butler Reginal Airport (\cite{Patrikar2021}).}
    \label{fig: flight distance}
\end{figure}

When the precise law for the non-stationary models before and after the change is known, the optimal algorithms for multi-stream QCD, for i.i.d. and non-i.i.d. settings, can be found in \cite{tart-book-2019, tart-niki-bass-2014}. When the pre-change law is known but the post-change law is unknown, optimal algorithms based on generalized likelihood ratios (GLR) or mixture are also available; see \cite{tart-book-2019, tart-niki-bass-2014}. However, 
GLR or mixture-based tests are computationally expensive to implement. 

In this paper, we take a robust approach to multi-stream quickest change detection leading to algorithms that are often computationally easy to implement in practice. We also show that our proposed algorithms are optimal for a robust problem formulation. The literature on robust QCD is reviewed in Section~\ref{sec:literature}. 

Our specific contributions are as follows:

\begin{enumerate}
    \item We first study the problem of robust QCD in a single stream where both pre- and post-change distributions are non-stationary and unknown. 
    We show that if there exists a pair of pre- and post-change non-stationary distributions, least favorable in a well-defined sense, then the test designed using this pair is robust optimal (see Section~\ref{sec: robust QCD model}). The least favorable distributions are identified using stochastic boundedness and monotone likelihood ratio conditions that are different from the standard conditions used in QCD literature (e.g., compare with \cite{unni-etal-ieeeit-2011}). Our conditions differ from the standard conditions in the literature because our post-change model depends on the change point. 
    \item We then study the multi-stream QCD problem with unknown non-stationary pre- and post-change distributions in each stream. In this setting, we obtain an optimal robust solution using the least favorable pair of distributions for each stream. We show that our algorithm can robustly detect the change when the change occurs in any arbitrary subset of the streams (see Section~\ref{sec: multi stream}). 
    \item We provide several analytical examples in which the least favorable pair of distributions can be identified (see Section~\ref{sec:LFDexamples}). 
    \item We apply our tests to real and simulated data and show that the robust tests can detect changes for all possible post-change scenarios (see Section~\ref{sec:NumericalResults}). For application to real datasets, we use COVID-19 infection rate data across different cities in the US to detect the onset of the pandemic. We also apply our algorithms to airline data collected around Pittsburgh airport for approaching aircraft. 
\end{enumerate}

\subsection{Related Works on QCD with Non-Stationary or Unknown Statistical Models}
\label{sec:literature}
In the QCD literature, optimal algorithms are developed when the post-change distribution is unknown under three categories: 1) generalized likelihood ratio (GLR) tests are developed where the unknown post-change parameter is replaced by its maximum likelihood estimate (\cite{lord-amstat-1971, lai-ieeetit-1998, tart-niki-bass-2014, lau2018binning, tart-book-2019}), 2) mixture-based tests are developed where a prior model is assumed on the post-change parameters and the likelihood ratio is integrated over this prior (\cite{lai-ieeetit-1998, poll-astat-1987, tart-book-2019, tart-niki-bass-2014}), and 3) robust tests are developed where the optimal tests are designed using a least favorable distribution (\cite{unni-etal-ieeeit-2011, oleyaeimotlagh2023quickest}). Of these three approaches, only the robust approach leads to statistics that can be calculated recursively.

In terms of non-stationary post-change models, optimal algorithms are developed under non-i.i.d. setting (\cite{lai-ieeetit-1998, tart-veer-siamtpa-2005, tart-niki-bass-2014, tart-book-2019, jain2020algorithms}), and under independent setting (\cite{liang2022quickest, brucks2023modeling}). While GLR and mixture approaches have been taken in these works, a result with a robust approach is not available in the literature. 

Exactly optimal solutions for unknown post-change distributions are developed in (\cite{unni-etal-ieeeit-2011}). However, their post-change model is i.i.d., and hence stationary. 
As for non-stationary settings, exactly optimal algorithms are developed in (\cite{bane-tit-2021}), but the post-change model is assumed to be known.
The post-change assumption can be relaxed to a least favorable distribution, and exact robust optimality for non-stationary processes is studied in (\cite{molloy2018minimax}). However, the notion of the least favorable distribution is stationary with time.
Both non-stationary and unknown distributions are considered in (\cite{oleyaeimotlagh2023quickest}), but the non-stationary model is limited to a statistically periodic model.
A broader (almost arbitrary) class of non-stationary processes is considered in (\cite{hou2024robust}), and exactly optimal tests in non-stationary and unknown distribution settings are developed. Nonetheless, the pre-change model is assumed to be known.

As discussed in the introduction, the QCD problem in multistream or distributed settings has been extensively studied. We refer to (\cite{tart-book-2019, tart-niki-bass-2014, veer-bane-elsevierbook-2013}) for a literature review and to (\cite{xie2013sequential, xu2022active, xu2021optimum, xu2021multi, fellouris2017multistream, wang2015large, oleyaeimotlagh2023quickest}) for some recent results. As with the single-stream case, while unknown and non-stationary post-change have been studied in the literature, a solution in the robust setting at the generality of our paper is not available in the literature. 

\section{A Classical QCD Model and Algorithm}
\label{sec:QCD model}
This section reviews a classical result in the QCD literature and discusses its limitations. 

In the QCD problem with i.i.d. observations, a decision maker observes a sequence of random variables $\{X_n\}$. Before a time $\nu$ (called the change point), the random variables are i.i.d. with a fixed density $f$, and after $\nu$, are i.i.d. with another density $g$:
\begin{equation}\label{eq:iidQCD}
	X_n \sim
	\begin{cases}
		f, &\quad \forall n < \nu, \\
		g, &\quad \forall n \geq \nu.
	\end{cases}
\end{equation}
Here $f$ and $g$ are densities such that
$$
D(g \; \| \; f) := \int g(x) \log \frac{g(x)}{f(x)} dx \; > \; 0,
$$
where $D(g\ \| \ f)$ is the Kullback-Leibler divergence between densities $g$ and $f$. 
The goal of the QCD problem is to detect this change in distribution from $f$ to $g$ with the minimum possible delay subject to a constraint on the rate of false alarms (\cite{poor-hadj-qcd-book-2009, veer-bane-elsevierbook-2013, tart-niki-bass-2014}). 

Studying this problem in the Bayesian setting would require us to assume a prior distribution on the change point $\nu$. Since such an assumption is not always practical, we focus on non-Bayesian settings in this paper. 
In the non-Bayesian settings, the change point $\nu$ is treated as an unknown constant (\cite{poll-astat-1985, lord-amstat-1971}). Let $\tau$ be a stopping time for the process $\{X_n\}$, i.e., a positive integer-valued random variable such that the event $\{\tau \leq n\}$ belongs to the $\sigma$-algebra generated by $X_1, \cdots, X_n$. In other words, whether or not $\tau \leq n$ is completely determined by the first $n$ observations. We consider this stopping time as our algorithm and declare that a change has occurred at the stopping time $\tau$. The detection delay of this stopping time 
can be measured using the conditional delay measures such as 
$$
\Expect_\nu \left[(\tau - \nu +1)^+\right] \quad \text{or} \quad \Expect_\nu \left[\tau - \nu | \tau \geq \nu \right], 
$$
where we use $\Prob_\nu$ to denote the law of the observation process $\{X_n\}$ when the change occurs at $\nu$. We use $\Expect_\nu$ to denote the corresponding expectation. In case of no change, notations $\Prob_\infty$ and $\Expect_\infty$ are used (suggesting a change occurring at $\nu=\infty$). However, these measures of delay displayed in the above equation depend on the unknown change point $\nu$. Thus, the popular measures for detection delay are minimax in nature. In the popular formulation of \cite{poll-astat-1985}, the minimax (minimizing the maximum of the conditional delay) metric considered is
$$
\sup_{\nu \geq 1} \; \Expect_\nu \left[\tau - \nu | \tau \geq \nu \right]. 
$$
However, this metric is not amenable to robust QCD theory in the generality we pursue in this paper. In this paper, we consider another popular minimax formulation proposed by \cite{lord-amstat-1971}:
\begin{equation}
\label{eq:LordenMetric}
    \textnormal{WADD}(\tau) \coloneqq \sup_\nu \; \esssup \; \Expect_\nu \left[(\tau - \nu +1)^+ | \mathscr{F}_{\nu-1}\right],
\end{equation}
where $\mathscr{F}_{\nu-1}$ is the $\sigma$-algebra generated by $X_1, \dots, X_{\nu-1}$, and $\esssup X$ represents the smallest constant $C$ such that $\Prob(X \leq C)=1$. Thus, the worst-case delay is taken over not just the change points but also over the past realizations. While this formulation may appear overly pessimistic due to the appearance of the essential supremum operation, it is the only metric for which strong optimality results have been developed in the literature in both minimax and robust settings (\cite{unni-etal-ieeeit-2011, hou2024robust}). 

The Lorden's formulation of the QCD problem is
\begin{equation}\label{eq:QCDproblem}
\begin{split}
    \inf_{\tau} \;\; 
    & \textnormal{WADD}(\tau),\\
    s.t. \;\; \
    & \Expect_{\infty}[\tau] \geq \frac{1}{\alpha},
\end{split}
\end{equation}
where $1/\alpha$ is a constraint on the mean time to a false alarm $\Expect_{\infty}[\tau]$, for $\alpha \in (0,1)$.
The optimal solution is given by the cumulative sum (CUSUM) algorithm  (\cite{page-biometrica-1954, mous-astat-1986, lai-ieeetit-1998}): 
\begin{equation}
    \label{eq:iidcusum}
    \begin{split}
            W_n &= \max_{1 \leq k \leq n} \sum_{i=k}^n \log \frac{g(X_i)}{f(X_i)},\\
            \tau_c &= \min\{n \geq 1: W_n \geq A\}.
    \end{split}
\end{equation}
To achieve optimality, the threshold $A$ in \eqref{eq:iidcusum} must be chosen carefully to meet a constraint on the false alarm with equality.
The CUSUM statistic $W_n$ also has an efficient recursion:
\begin{equation}
    W_n = \left(W_{n-1} + \log \frac{g(X_n)}{f(X_n)}\right)^+, \quad W_0=0.
\end{equation}
Deep insights can be obtained on the CUSUM algorithm through its asymptotic analysis. Since the CUSUM algorithm is optimal, it is also asymptotically optimal for the problem as $\alpha \to 0$. The following theorem provides the asymptotic performance of the CUSUM algorithm.  

\begin{theorem}
\label{thm:optimal them cusum}
    (Asymptotic optimality of CUSUM (\cite{lord-amstat-1971})):
    If $0 < D(g \ \| \ f) < \infty$, with $A = \log(1/\alpha)$,
    \begin{equation}
        \Expect_{\infty}[\tau_c] \geq \frac{1}{\alpha},
    \end{equation}
    and as $\alpha \to 0$,
    \begin{equation}
        \textnormal{WADD}(\tau_c) = \Expect_1[\tau_c - 1] = \frac{|\log \alpha|}{D(g \ \| \ f)} + o(1).
    \end{equation}
    Here the $o(1)$ term goes to zero as $\alpha \to 0$. 
    Because of the above equations (and can also be established directly), for any stopping time $\tau$ with $\Expect_{\infty}[\tau] \geq {1}/{\alpha}$,
    \begin{equation}
        \textnormal{WADD}(\tau) \geq \frac{|\log \alpha|}{D(g \ \| \ f)} + o(1), \quad \textnormal{ as } \alpha \to 0.
    \end{equation}
\end{theorem}


\section{Robust Single-Stream QCD with Unknown Non-Stationary Pre- and Post-Change Distributions}
\label{sec: robust QCD model}

In this section, we present our results on robust QCD in single stream data with unknown non-stationary pre- and post-change distributions. The results of the more general multi-stream case are discussed in Section~\ref{sec: multi stream}. While the result of this section on the single-stream case is a special case of the multi-stream result, we present it here separately to emphasize the proof technique and the role of different assumptions. Also, it is notationally easier to discuss examples and applications in the single stream case.  

\subsection{Change Point Model and Notations for Non-Stationary Data}
The classical QCD model (see \eqref{eq:iidQCD}) assumes that the observations $\{X_n\}$ are i.i.d. with fixed density $f$ before a change time $\nu$, and i.i.d. with density $g$ after $\nu$. 
The change point model that we study  is as follows: 
\begin{equation}\label{eq:changepointmodel}
	X_{n} \sim
	\begin{cases}
		f_{n}, &\quad \forall n < \nu, \\
		g_{n, \nu}, &\quad \forall n \geq \nu.
	\end{cases}
\end{equation}
The random variables $\{X_n\}$ are independent conditioned on the change point $\nu$. 
Thus, before the change time $\nu$, the process follows a sequence of densities $\{f_{n}\}$. After time $\nu$, the process changes to a sequence of densities $\{g_{n, \nu}\}$. Both the pre- and post-change distribution may depend on the observation time $n$. Besides, the post-change distribution may depend on the change point $\nu$. The densities are assumed to satisfy
\begin{equation}
\nonumber
	D(g_{n, \nu} \; \| \; f_{n}) > 0, \quad \forall n \geq 1, \;   \nu \geq 1, \; n \geq \nu. 
\end{equation}
We will refer to a process with pre-change densities $\{f_{n}\}$ and post-change densities $\{g_{n, \nu}\}$ as a non-stationary process since the density of the data evolves or changes with time. 

We use $F = \{f_{n}\}$ and $G = \{g_{n, \nu}\}$ to denote the pre- and post-change sequence of densities, respectively. In our model, $F=\{f_{n}\}$ and $G=\{g_{n, \nu}\}$ are not known to the decision-maker. However, we assume that there are families of distributions $\{\mathcal{P}^{F}_{n}\}$ and $\{ \mathcal{P}^{G}_{n, \nu}\}$ such that
\begin{equation*}
    (f_{n}, g_{n, \nu}) \in \mpp^{F}_{n} \times \mathcal{P}^{G}_{n, \nu}, \quad n, \nu = 1,2, \dots, 
\end{equation*}
and the families $\{\mathcal{P}^{F}_{n}\}$ and $\{ \mathcal{P}^{G}_{n, \nu}\}$ are known to the decision-maker. 
We use the notation 
\begin{equation*}
\mf = \{F = \{f_{n}\}: f_{n} \in \mpp^{F}_{n}, \; n \geq 1\}
\end{equation*}
to denote the set of all possible pre-change laws, and
\begin{equation*}
    \mathcal{G} = \{G = \{g_{n, \nu}\}: g_{n, \nu} \in \mathcal{P}^{G}_{n, \nu}, \; n, \nu \geq 1\}
\end{equation*}
to denote the set of all possible post-change laws. 
We use $\Prob_{\nu}^{F, G}$ to denote the law of the observation process $\{X_{n}\}$ when the change occurs at $\nu$ and the pre-change and post-change laws are given by $F$ and $G$, respectively. We use $\Expect_{\nu}^{F, G}$ to denote the corresponding expectation. 
The notations $\Prob_{\infty}^{F, G}=\Prob^{F}_{\infty}$ and $\Expect_{\infty}^{F, G}=\Expect^{F}_{\infty}$ are used when there is no change (change occurring at $\nu=\infty$), as the post-change law plays no role when the change never occurs. 

\subsection{Problem Formulation for Single-Stream Robust QCD}
In the following, we use Lorden's criteria (\cite{lord-amstat-1971}) to find the best stopping time to detect the change in distribution in non-stationary processes in a robust manner. We define 
\begin{equation}
\label{eqn:wadd definition}
    \text{WADD}^{F, G}(\tau) \coloneqq \sup_\nu \; \esssup \; \Expect_\nu^{F, G}\left[(\tau - \nu + 1)^+ | \mathscr{F}_{\nu-1}\right],
\end{equation}
where $\mathscr{F}_{n} = \sigma(X_{1}, \dots, X_{n})$. This is Lorden's delay metric when the pre-change law is $F$ and the post-change law is $G$. 
Since the decision maker is unaware of the exact pre- and post-change laws, we seek a stopping time to minimize the following worst-case version of the detection delay:
\begin{equation}\label{eq:robustProbmini}
\begin{split}
    \inf_{\tau}  \sup_{F \in \mf, G \in \mg}   
    &\text{WADD}^{F, G}(\tau)\\
 s.t. \quad &\inf_{F \in \mf}\Expect^{F}_\infty[\tau] \geq \frac{1}{\alpha}.
\end{split}
\end{equation}
{We say that a solution is robust optimal if a stopping time or algorithm is a solution to the problem in \eqref{eq:robustProbmini}}. 

\subsection{Robust Algorithm for QCD in Non-Stationary Data}
 We now provide the optimal or asymptotically optimal solution to \eqref{eq:robustProbmini} under assumptions on the families of pre-change uncertainty classes $\{\mathcal{P}^{F}_{n}\}$ and post-change uncertainty classes $\{\mathcal{P}^{G}_{n, \nu}\}$. 
We assume in the rest of this section that all densities involved are equivalent to each other (absolutely continuous with respect to each other). 

To state the assumptions on $\{\mpp^{F}_n\}$ and $\{\mathcal{P}^{G}_{n, \nu}\}$, we need some definitions. We say that a random variable $Z_2$ is stochastically larger than another random variable $Z_1$ (and write $
Z_2 \succ Z_1
$) if
$$
\Prob(Z_2 \geq t) \geq \Prob(Z_1 \geq t), \quad \forall t \in \mathbb{R}.
$$
We now introduce the notion of stochastic boundedness in non-stationary processes. 

\begin{definition}[Stochastic Boundedness and Least Favorable Laws (LFLs) in Non-Stationary Processes]
\label{def1}
We say that the pair of families ($\{\mathcal{P}^{F}_{n} \}, \{\mathcal{P}^{G}_{n, \nu}\}) $ are stochastically bounded by least favorable laws (LFLs) $(\bar{F}, \bar{G}) = (\{\bar{f}_{n}\}, \{\bar{g}_{n, \nu}\})$, with $\bar{f}_{n} \in \mathcal{P}^{F}_{n}$ and 
$
\bar{g}_{n, \nu} \in \mathcal{P}^{G}_{n, \nu}, \; n, \nu=1,2, \dots,
$
or that the pair $(\bar{F}, \bar{G})$ is least favorable if 
\begin{equation}\label{eq:stocbounded-1}
		\begin{split}
			 \bar f_{n}
            &\succ  {f}_{n}, \quad \text{ for every } \; f_{n} \in \mathcal{P}^{F}_{n},\quad n = 1,2, \dots, 
		\end{split}
	\end{equation}
and 
\begin{equation}\label{eq:stocbounded-2}
        \begin{split}
			 {g_{n, \nu}} 
            &\succ 
			\bar{g}_{n, \nu}, \quad \text{ for every } \; g_{n, \nu} \in \mathcal{P}^{G}_{n, \nu}, \quad n,\nu=1,2, \dots.       
		\end{split}
    \end{equation}
\end{definition}

\begin{remark}
    We note that the conditions above differ from those used, e.g., in \cite{unni-etal-ieeeit-2011} and \cite{oleyaeimotlagh2023quickest} because in our case the post-change law depends on the change point. 
\end{remark}

Given a pair of LFLs, we use the following algorithm for robust QCD: 
\begin{align} \label{eq:W_t_n}
     \bar{W}_{n} = \max_{1 \leq k \leq n} \sum_{i=k}^n \log \frac{\bar{g}_{i,k}(X_{i})}{\bar f_{i}(X_{i})},
\end{align}
with stopping rule
\begin{align} \label{eq:tau_t_star}
    \tau_{ss} = \inf\{n \geq 1 : \bar W_{n} \geq \bar A_{n, \alpha}\}.
\end{align}
Here thresholds $ \bar{A}_{n,\alpha}$ are chosen to satisfy the constraint on the mean time to a false alarm. The statistic $\bar{W}_{n}$ is the classical CUSUM statistic \eqref{eq:iidcusum} modified for the non-stationary setting and based on the LFLs. 

The optimality of $\tau_{ss}$ when the true pre- and post-change laws of $\{X_n\}$ are given by the LFLs has been studied in \cite{lai-ieeetit-1998, tart-niki-bass-2014, tart-book-2019, liang2022quickest, brucks2023modeling, bane-tit-2021, oleyaeimotlagh2023quickest, bane-elsevier-2024}.  In these works, it has been shown that under mild assumptions on the non-stationary models $\bar{F}$ and $\bar{G}$,  the stopping rule $\tau_{ss}$ is asymptotically optimal (as $\alpha \to 0$) when the pre- and post-change laws are also given by $\bar{F}$ and $\bar{G}$, respectively. In Section~\ref{sec:designSingleStream} below, we provide a brief discussion on some standard sufficient conditions on optimality.  

\subsection{Robust Optimality in a Non-Stationary Setting}
In this section, we prove the robust optimality of $\tau_{ss}$ for the problem in \eqref{eq:robustProbmini}. 
Specifically, we show in  Theorem~\ref{thm:LFLRobust_minimax} that if the stopping rule $\tau_{ss}$ is optimal for \eqref{eq:QCDproblem}
when the pre- and post-change laws are $\bar F$ and $\bar G$, respectively, then it is robust optimal for the problem in \eqref{eq:robustProbmini}. For the theorem, we need the following lemma from (\cite{unni-etal-ieeeit-2011}), which we reproduce here for readability.

\begin{lemma}[Stochastic dominance (\cite{unni-etal-ieeeit-2011})]
\label{lem:stocbound_UV}
    Suppose $\{U_i: 1 \leq i \leq n\}$ is a set of mutually independent random variables, and $\{V_i: 1 \leq i \leq n\}$ is another set of mutually independent random variables such that $U_i \succ V_i$, $1 \leq i \leq n$. Now let $q: \mathbb{R}^n \to \mathbb{R}$ be a continuous real-valued function defined on $\mathbb{R}^n$ that satisfies
    \begin{align*}
        q(X_1, \dots, X_{i-1}, a, X_{i+1}, \dots, X_{n}) \geq q(X_1, \dots, X_{i-1}, X_i, X_{i+1}, \dots, X_{n})
    \end{align*}
    for all $(X_1, \dots, X_n) \in \mathbb{R}^n$, $a > X_i$, and $i \in \{1, \dots, n\}$. Then we have 
    \begin{align*}
        q(U_1, U_2, \dots, U_n) \succ q(V_1, V_2, \dots, V_n).
    \end{align*}
\end{lemma}

\begin{theorem}
\label{thm:LFLRobust_minimax}
Suppose the following conditions hold:
    \begin{enumerate}
	\item[(a)] 	The pair of families ($\{\mathcal{P}^{F}_{n} \}, \{\mathcal{P}^{G}_{n, \nu}\}) $ have a pair of LFLs $(\bar{F}, \bar{G}) = (\{\bar{f}_{n}\}, \{\bar{g}_{n, \nu}\})$ as defined in Definition \ref{def1}. 
	\item[(b)] Let $\alpha \in (0,1)$ be a constraint on the false alarm. Let the thresholds $\{A_{n, \alpha}\}$ used by the stopping rule $\tau_{ss}$ be such that $\Expect^{\bar{F}}_\infty[\tau_{ss}] = {1}/{\alpha}$, where $\tau_{ss}$ in \eqref{eq:tau_t_star} is the optimal rule designed using the LFL.
	\item[(c)] All likelihood ratio functions used in the stopping rule $\tau_{ss}$ are continuous and monotone increasing.
  \end{enumerate}
Then the following results are true:
    \begin{enumerate}
    \item If the stopping rule $\tau_{ss}$ in \eqref{eq:tau_t_star} is exactly optimal for the problem in  \eqref{eq:QCDproblem} when the pre- and post-change laws are given by the LFLs $\bar{F}$ and  $\bar{G}$, respectively, then the stopping rule is exactly robust optimal for the problem in \eqref{eq:robustProbmini}.
    \item If the stopping rule $\tau_{ss}$ in \eqref{eq:tau_t_star} is asymptotically optimal, as $\alpha \to 0$,  for the problem in  \eqref{eq:QCDproblem} when the pre- and post-change laws are given by the LFLs $\bar{F}$ and  $\bar{G}$, respectively, then the stopping rule is asymptotically robust optimal for the problem in \eqref{eq:robustProbmini}, as $\alpha \to 0$. 
    \end{enumerate}
	
\end{theorem}

\begin{proof}
See Appendix \ref{appen:1}.
\end{proof}

\subsection{On Optimality of the Robust Algorithm under LFLs}
\label{sec:designSingleStream}
One of the main conditions for robust optimality of the stopping rule $\tau_{ss}$ in \eqref{eq:tau_t_star} is that it be optimal or asymptotically optimal when the pre- and post-change laws are given by the LFLs. In this section, we briefly discuss conditions under which such an optimality can be attained. 


In the theorem below, we provide some general conditions on the asymptotic optimality of $\tau_{ss}$ under the LFLs. For exact optimality in a non-stationary setting, we refer the readers to \cite{bane-tit-2021}. For conditions on asymptotic optimality more general than those given in the theorem below, we refer the readers to \cite{lai-ieeetit-1998, tart-veer-siamtpa-2005, tart-niki-bass-2014, tart-book-2019, liang2022quickest}. 

Define
$$
\bar Z_{n, \nu} = \log \frac{\bar{g}_{n,\nu}(X_n)}{\bar f_n(X_n)}
$$
to be the log-likelihood ratio at time $n$ calculated for the LFLs when the change occurs at $\nu$. The theorem below gives sufficient conditions on the sequence $\{\bar Z_{n, \nu} \}_{n, \nu}$ for $\tau_{ss}$ to be optimal under the LFLs.
\begin{theorem}[\cite{brucks2023modeling},\cite{lai-ieeetit-1998}]
	\label{thm:modifiedconds}
	\begin{enumerate}
		\item Let there exist a positive number $I$ such that the log-likelihood ratios $\{\bar Z_{n, \nu}\}$ satisfy the following  condition: for any $\delta > 0$, 
		\begin{equation}
			\label{eq:Znnu_LB}
			\begin{split}
				\lim_{n \to \infty} \; \sup_{\nu \geq 1} \; \esssup \mathsf{P}_\nu^{\bar F, \bar{G}} &\left(\max_{t \leq n} \sum_{i = \nu }^{\nu + t} \bar Z_{i, \nu} \geq I(1+\delta)n \; \bigg| \; X_1, \dots, X_{\nu-1}\right) = 0.
			\end{split}
		\end{equation}
		Then, we have the universal lower bound as $\alpha \to 0$, 
		\begin{equation}
			\begin{split}
				\inf_{\tau} \textup{WADD}^{\bar F, \bar{G}}(\tau) \; \geq \; \frac{|\log \alpha|}{I} (1+o(1)). 
			\end{split}
		\end{equation}
		Here the minimum over $\tau$ is over those stopping times satisfying $\Expect_\infty^{\bar{F}}[\tau] \geq \frac{1}{\alpha}$. 
		\item The robust algorithm $\tau_{ss}$ in \eqref{eq:tau_t_star}, when we choose the threshold $\bar A_{n, \alpha} = |\log \alpha|$,
		$$
		\tau_{ss} = \min\left\{n \geq 1: \max_{1 \leq k \leq n} \sum_{i=k}^n \bar Z_{i,k} \geq |\log \alpha|\right\},
		$$
		satisfies
		$$
		\mathsf{E}_\infty^{\bar F}[\tau_{ss}] \geq \frac{1}{\alpha}.
		$$
		\item Furthermore, if the log-likelihood ratios $\{\bar Z_{n, \nu}\}$ also satisfy, $\forall \delta > 0$, 
		\begin{equation}
			\label{eq:Znnu_UB}
			\begin{split}
				\lim_{n \to \infty} \; \sup_{k \geq \nu \geq 1} \; \esssup \mathsf{P}_\nu^{\bar F, \bar{G}} & \left(\frac{1}{n}\sum_{i = k }^{k +n} \bar Z_{i,k} \leq I - \delta \; \bigg| \; X_1, \dots, X_{k-1}\right) = 0.
			\end{split}
		\end{equation}
		Then as $\alpha \to 0$, $\tau_{ss}$ achieves the lower bound:
		\begin{equation}
			\begin{split}
				\textup{WADD}^{\bar F, \bar{G}}(\tau_{ss})  \leq  \frac{|\log \alpha |}{I}(1+o(1)), \quad \alpha \to 0. 
			\end{split}
		\end{equation}
	\end{enumerate}
\end{theorem}

Some special cases and well-studied models where the conditions of the theorem are satisfied are given below. 
\begin{enumerate}
    \item \textit{I.I.D. process}: In the special case when
\begin{equation}
    \label{eq:robustexactiid_2}
    \bar f_n = \bar f, \quad \bar{g}_{n, \nu} = \bar{g}, \quad \forall n, \nu,
\end{equation}
the robust algorithm given above is the robust CUSUM algorithm designed using densities $\bar f$ and $\bar g$. All the conditions of the theorem above are satisfied with $I = D(\bar{g} \; \| \; \bar f)$, the Kullback-Leibler divergence between $\bar g$ and $\bar f$. 
\item \textit{I.P.I.D. process}: Consider the case when there exist a positive integer $T$ such that
\begin{equation}
    \bar f_{n+T} = \bar f_n, \quad \bar{g}_{n, \nu} = \bar{g}_n, \quad \bar g_{n+T}=\bar g_n.
\end{equation}
A process with this periodic sequence of densities is called an independent and periodically identically distributed process (\cite{bane-tit-2021, oleyaeimotlagh2023quickest, bane-elsevier-2024}).  
The conditions of Theorem~\ref{thm:modifiedconds} are satisfied with
    $$
    I = \frac{1}{T} \sum_{n=1}^T D(\bar g_n \; \| \; \bar f_n). 
    $$
    \item \textit{MLR-order processes}: 
Finally, consider the case where we have a sequence of densities $\{\bar{h}_n\}$ such that
\begin{equation}
    \bar f_{n+T} = \bar f, \quad \bar{g}_{n, \nu} = \bar{h}_{n-\nu+1}. 
\end{equation}
In addition, 
  the densities $\{\bar{h}_n\}$ are increasing in MLR order: for all $n \geq 1$, 
    $$
    \frac{\bar{h}_{n+1}(x)}{\bar{h}_n(x)} \quad \uparrow \quad x,
    $$
    that is, the likelihood ratio is monotonically increasing in $x$. Such a process is called an exploding process or an monotone likelihood ratio (MLR)-order process (\cite{brucks2023modeling}).
    The conditions of Theorem~\ref{thm:modifiedconds} are satisfied with
    $$
    I = \lim_{N \to \infty} \frac{1}{N} \sum_{n=1}^N D(\bar{h}_n \; \| \; \bar f),
    $$ 
    provided the above limit exists (\cite{brucks2023modeling}). 
\end{enumerate}




\subsection{On Performance of the Robust Algorithm under Laws other than LFLs}
In this section, we discuss the false alarm and delay performances of the robust algorithm when the pre- and post-change laws are not given by the LFLs. The fact that the robust algorithm is designed using the LFLs helps us obtain bounds on the performances. The proof of the following corollary follows from the proof steps used for Theorem~\ref{thm:LFLRobust_minimax}. 

\begin{corollary}
    Suppose the log-likelihood ratios $\{\bar Z_{n, \nu}\}$ satisfy the conditions in Theorem~\ref{thm:modifiedconds}. Then the robust algorithm $\tau_{ss}$ in \eqref{eq:tau_t_star}, when we choose the threshold $\bar A_{n, \alpha} = |\log \alpha|$,
		$$
		\tau_{ss} = \min\left\{n \geq 1: \max_{1 \leq k \leq n} \sum_{i=k}^n \bar Z_{i,k} \geq |\log \alpha|\right\},
		$$
		satisfies
        \begin{equation}
			\begin{split}
            \mathsf{E}_\infty^{ F}[\tau_{ss}] &\geq \frac{1}{\alpha}, \quad \forall F \in \mathcal{F}. \\
				\textup{WADD}^{ F, {G}}(\tau_{ss})  &\leq  \frac{|\log \alpha |}{I}(1+o(1)), \quad \alpha \to 0, \quad \forall F \in \mathcal{F}, G \in \mathcal{G}. 
			\end{split}
		\end{equation}
\end{corollary}
Thus, no matter what the pre-change non-stationary law $F$ is, if the thresholds of the stopping rule $\tau_{ss}$ is set to a constant threshold $|\log \alpha|$, the mean time to a false alarm for $\tau_{ss}$  is always lower bounded by $\frac{1}{\alpha}$. Also, the detection delay is asymptotically upper bounded by $\frac{|\log \alpha |}{I}$ for any non-stationary pre- and post-change laws $F$ and $G$. 

\subsection{Examples of Least Favorable Laws}
\label{sec:LFDexamples}
We now provide some specific examples of LFLs from the Gaussian and Poisson families. 
We use the following simple lemma in these examples. The lemma is stated for densities but a similar result is true also for mass functions. Its proof can be found in (\cite{krishnamurthy2016partially}). We provide the proof for completeness. 
\begin{lemma}
\label{lem:MLRorder}
    Let $f$ and $g$ be two probability density functions such that $g$ dominated $f$ in monotone likelihood ratio (MLR) order:
    $$
    \frac{g(x)}{f(x)} \; \; \uparrow \; \; x. 
    $$
    Then, $g$ also dominates $f$ in stochastic order, i.e.,
    $$
    \int_{x}^\infty g(y) dy \geq \int_x^\infty f(y) dy, \quad \forall x. 
    $$
\end{lemma}
\begin{proof}
    Define $t = \sup \left\{x: \frac{g(x)}{f(x)} \leq 1\right\}$.
    Then, when $x \leq t$, 
    $$
    \int_{-\infty}^x g(y) dy \leq  \int_{-\infty}^x f(y) dy,
    $$
    implying $
    \int_{x}^\infty g(y) dy \geq \int_x^\infty f(y) dy.
    $
    On the other hand, if $x > t$, then 
    $$
    \int_{x}^\infty g(y) dy =\int_{x}^\infty \frac{g(y)}{f(y)}f(y) dy \geq \int_x^\infty f(y) dy.
    $$
    This proves the lemma. 
\end{proof}

We now give two examples of non-stationary processes with specified uncertainty classes. We then identify their LFLs. 

\begin{example}(Gaussian LFL)
\label{exam:GaussLFD}
Let the pre- and post-change density be given by
\begin{align*}
    f_n = \mathcal{N}(\theta_{n}, 1), \quad
    g_{n, \nu} = \mathcal{N}(\mu_{n,\nu}, 1). 
\end{align*}
The parameters $\{\theta_n\}$ and $\{\mu_{n,\nu}\}$ are not known but are believed to satisfy for some known $\{\bar\theta_n\}$ and $\{\bar{\mu}_{n,\nu}\}$, respectively,
\begin{align*}
    \mu_{n,\nu} \geq \bar{\mu}_{n,\nu}, \quad
    \bar\theta_n \geq \theta_n, \quad \forall n, \text{ and } \forall \nu \leq n,
\end{align*}
such that for each $n$ and for all $\nu \leq n$,
$$
\bar{\theta}_n < \bar{\mu}_{n, \nu}. 
$$
The parameters $\{\bar\theta_n\}$ and $\{\bar{\mu}_{n, \nu}\}$ are known to the decision-maker. 

The following likelihood ratio calculations are straightforward:
\begin{align*}
\log \frac{\bar{g}_{n, \nu}(X)}{\bar f_n(X)} 
&=  \left(\bar{\mu}_{n,\nu} - \bar\theta_n\right) X - \frac{\bar{\mu}_{n,\nu}^2 - \bar\theta_n^2}{2}, \\
\log \frac{{g}_{n, \nu}(X)}{\bar{g}_{n, \nu}(X)} 
&=  \left(\mu_{n,\nu} - \bar{\mu}_{n,\nu}\right) X - \frac{{\mu}_{n,\nu}^2 - \bar{\mu}_{n,\nu}^2}{2},\\
\log \frac{\bar f_n(X)}{ f_n(X)} 
&=  \left( \bar\theta_n - \theta_n\right) X - \frac{\bar\theta_n^2 - \theta_n^2}{2}.
\end{align*}
Note that all the likelihood ratios are monotonically increasing and continuous in $X$, for any choice of $n$ and $\nu \leq n$. By Lemma~\ref{lem:MLRorder}, this means that 
$$
{g}_{n, \nu} 
\succ \bar{g}_{n, \nu}, \quad \text{ and } \quad \bar f_n \succ f_n. 
$$
Thus, the following densities are LFLs:
\begin{align*}
    \bar{f}_n = \mathcal{N}(\bar{\theta}_{n}, 1), \quad
    \bar{g}_{n, \nu} = \mathcal{N}(\bar{\mu}_{n,\nu}, 1), \quad n \geq 1, \quad \nu \leq n. 
\end{align*}

\end{example}

\begin{example}(Poisson LFL)
\label{exam:PoissonLFD}
Let the pre- and post-change density be given by
\begin{align*}
f_n =  \text{Pois}(\gamma_n), \quad g_{n, \nu} = \text{Pois}(\lambda_{n,\nu}).        
\end{align*}
where $\max_n \{\gamma_n\} \leq \min_{n, \nu}\{\lambda_{n, \nu}\}$.
$\{\gamma_n\}$ and $\{\lambda_{n,\nu}\}$ are not known but are believed to satisfy for some known $\{\bar\gamma_n\}$ and $\{\bar{\lambda}_{n,\nu}\}$, respectively,
\begin{align*}
    \lambda_{n,\nu} \geq \bar{\lambda}_{n,\nu}, \quad 
    \bar\gamma_n \geq \gamma_n, \quad \forall n, \nu,
\end{align*}
and for each $n$ and $\nu \leq n$,
$$
\bar{\gamma}_n < \bar{\lambda}_{n, \nu}. 
$$
Then,
\begin{align*}
\log \frac{\bar{g}_{n, \nu}(X)}{\bar f_n(X)} 
&=   \log \left[\left(\frac{\bar{\lambda}_{n,\nu}}{\bar\gamma_n}\right)^X e^{-\bar{\lambda}_{n,\nu} + \bar\gamma_n}\right] 
= X \log \left(\frac{\bar{\lambda}_{n,\nu}}{\bar\gamma_n}\right) -\bar{\lambda}_{n,\nu} + \bar\gamma_n,\\
\log \frac{{g}_{n, \nu}(X)}{\bar {g}_{n, \nu}(X)} 
&=   \log \left[\left(\frac{{\lambda}_{n,\nu}}{\bar{\lambda}_{n,\nu}}\right)^X e^{-{\lambda}_{n,\nu} + \bar{\lambda}_{n,\nu}}\right] 
= X \log \left(\frac{{\lambda}_{n,\nu}}{\bar{\lambda}_{n,\nu}}\right) -{\lambda}_{n,\nu} + \bar{\lambda}_{n,\nu},\\
\log \frac{\bar f_{n}(X)}{ f_n(X)} 
&=   \log \left[\left(\frac{\bar\gamma_n}{\gamma_n}\right)^X e^{-\bar\gamma_n + \gamma_n}\right] 
= X \log \left(\frac{\bar\gamma_n}{\gamma_n}\right) -\bar\gamma_n + \gamma_n.
\end{align*}
Note that because of the assumptions made on the order of the means, all the likelihood ratios are monotonically increasing and continuous in $X$, for any choice of $n$ and $\nu \leq n$. By Lemma~\ref{lem:MLRorder}, this means that 
$$
{g}_{n, \nu} 
\succ \bar{g}_{n, \nu}, \quad \text{ and } \quad \bar f_n \succ f_n. 
$$
Thus, the following densities are LFLs:
\begin{align*}
    \bar{f}_n =\text{Pois}(\bar{\gamma}_n), \quad \bar{g}_{n, \nu} = \text{Pois}({\bar\lambda}_{n,\nu}).     \quad n \geq 1, \quad \nu \leq n. 
\end{align*}

\end{example}

\section{Robust Multi-Stream QCD with Unknown Non-Stationary Pre- and Post-Change Distributions}
\label{sec: multi stream}

In this section, we obtain a robust solution to a multi-stream QCD problem where data in each stream is non-stationary, and the pre- and post-change distributions in each stream are not known to the decision-maker. At the change point, it is assumed that the change affects a class of subsets of streams. 


\subsection{Multi-Stream Change Point Model and Notations}
We first state the change point model and the problem formulation. 
Suppose $\{X_{\theta, n}\}$ is an independent sequence of random variables observed over time $n$ at stream indexed by $\theta \in \Theta =\{\theta_1, \dots, \theta_M\}$. Let $B \subset \{\theta_1, \dots, \theta_M\}$. 
We assume that changes occur simultaneously in streams indexed by elements of $B$, which is unknown to the decision maker. Specifically, if $\theta \in B$, then the process $\{X_{\theta, n}\}$ follows a sequence of densities $\{f_{\theta, n}\}$ before the change time $\nu$. After time $\nu$, the process changes to a sequence of densities $\{g_{\theta, n, \nu}\}$. 
Mathematically,  
\begin{equation}\label{eq:changepointmodel B}
	X_{\theta,n} \sim
	\begin{cases}
		f_{\theta, n}, &\quad \forall n < \nu, \\
		g_{\theta, n, \nu}, &\quad \forall n \geq \nu,
	\end{cases} \quad \text{if} \quad \theta\in B. 
\end{equation}
If $\theta \notin B$, the density sequence of the observations remains equal to $f_{\theta, n}$,
$$
X_{\theta, n} \sim f_{\theta, n}, \quad \forall n, \quad \text{if} \quad \theta \notin B. 
$$
As in the single-stream case, the densities before and after change are not known to the decision-maker. 
While the subset  $B \subset \Theta$ of streams, where the change occurs, is not known to the decision maker, we assume that there is a known collection $\mb$ of subsets of $\Theta$ such that 
$B \in \mb$. For example, we can have $\mb$ as the collection of all subsets of size less than a specified $K$: 
$$
\mb = \{B: |B| \leq K\}. 
$$
Typically, we would have $
K \ll M. 
$
At the end of this section, we study the special case of $K = 1$, i.e., the case where the change affects only one of the streams. 

For each $\ta\in \Theta$, the densities are assumed to satisfy
\begin{equation}
\nonumber
	D(g_{\theta, n, \nu} \; \| \; f_{\theta, n}) > 0, \quad \forall n \geq 1, \;   \nu \geq 1, \; n \geq \nu. 
\end{equation}
We refer to a process with pre-change $\{f_{\theta, n}\}$ and post-change $\{g_{\theta, n, \nu}\}$ as a non-stationary process, since the density of the data changes with time among streams.

Denote the non-stationary models before and after the change in stream $\theta$ as
\begin{equation*}
    F^\theta = \{f_{\theta, n}\}, \quad G^{\theta} = \{g_{\theta, n, \nu}\}, \quad \theta \in \Theta,
\end{equation*}
and the family of pre- and post-change models across the streams as
$$
F = \{F^\theta\}_{\theta \in \Theta}, \quad \text{ and } \quad G = \{G^\theta\}_{\theta \in \Theta}. 
$$
While the laws $F$ and $G$ are not known, we assume that there are known families of densities $\mathcal{P}^{F, \theta}_{n}$ and $\mathcal{P}^{G, \theta}_{n, \nu}$ such that
\begin{equation*}
    f_{\ta, n} \in \mathcal{P}^{F, \theta}_{n} , \quad g_{\theta, n, \nu} \in \mathcal{P}^{G, \theta}_{n, \nu}, \quad n,\nu = 1,2, \dots. 
\end{equation*}
Furthermore, define 
\begin{align}
\label{eq: dist family}
\mf &= \{F = \{F^\theta\}_{\theta \in \Theta}: F^\theta = \{f_{\theta, n}\}, \; f_{\theta, n} \in \mpp^{F, \theta}_{n}, \; n \geq 1, \; \theta \in \Theta\},\\
\mathcal{G} &= \{G = \{G^{\theta}\}_{\theta \in \Theta}: G^{\theta} = \{g_{\theta, n, \nu}\}, \; g_{\theta, n, \nu} \in \mathcal{P}^{G, \theta}_{n, \nu}, \; n, \nu \geq 1, \; \theta \in \Theta\}
\end{align} 
to be the collection of all possible pre- and post-change non-stationary process laws across all the streams. 

Finally, we use the notation
$$
\Prob_{\nu}^{F, G, B}
$$
to denote the probability measure such that change occurs at time $\nu$ in the subset $B$ when the pre- and post-change non-stationary models are $F$ and $G$, respectively. Also, let 
$\Expect_{\nu}^{F, G, B}$ be the corresponding expectation. We use the notations $\Prob_{\infty}^{F, G, B}=\Prob^{F}_{\infty}$ and $\Expect_{\infty}^{F, G, B}=\Expect^{F}_{\infty}$ when there is no change.

\subsection{Problem Formulation for Multi-Stream QCD with Non-Stationary Data}
Define the Lorden detection delay metric (\cite{lord-amstat-1971}) under the new notations as
\begin{equation}
\label{eqn:WADD2}
    \text{WADD}^{F, G, B}(\tau) \coloneqq \sup_\nu \; \esssup \; \Expect_\nu^{F, G, B}\left[(\tau - \nu+1)^+ | \mathscr{F}_{\nu-1}\right],
\end{equation}
where $\mathscr{F}_{n} = \sigma(\{X_{\theta, 1}, \dots, X_{\theta, n}\}_{\theta \in \Theta})$.
When the laws $F$ and $G$ are known, and streams $B$ are also known, the classical Lorden's formulation for this problem can be stated as (see \eqref{eq:QCDproblem}):
\begin{equation}\label{eq:QCDproblem-B}
\begin{split}
    \inf_{\tau} \;\; 
    & \textnormal{WADD}^{F, G, B}(\tau),\\
    s.t. \;\; \
    & \Expect_{\infty}^F[\tau] \geq \frac{1}{\alpha}.
\end{split}
\end{equation}

Since the laws $(F, G)$ and stream $B$ are not exactly known, we study a robust version of the minimax problem for arbitrary $B \in \mathcal{B}$:
\begin{equation}\label{eqn:robustprob2}
\begin{split}
  \inf_{\tau} \;\; {\sup_{F \in \mf, G \in \mathcal{G}}} \; \; 
  &\text{WADD}^{F, G, B}(\tau)\\
  s.t. \;\; \
    & \inf_{F \in \mf}\Expect^{F}_{\infty}[\tau] \geq \frac{1}{\alpha}.
\end{split}
\end{equation}
We say that a solution is robust optimal if a stopping time or algorithm is a solution to the problem in \eqref{eqn:robustprob2}.
As suggested by the formulation, we seek a solution that is robust optimal uniformly across every $B \in \mathcal{B}$. 

\subsection{Robust Algorithm for QCD in Multi-Stream Non-Stationary Processes}
We provide the asymptotically optimal solution to \eqref{eqn:robustprob2} under the assumption that the families of pre-change uncertainty classes $\{\mathcal{P}^{F, \ta}_{n}\}_{n, \ta}$ and post-change uncertainty classes $\{\mathcal{P}^{G, \ta}_{n, \nu}\}_{n, \ta}$ have LFLs for each $\theta$. The definition of LFLs is similar to Defintion~\ref{def1}, but we provide it here for completeness.

\begin{definition}[Stochastic Boundedness and Least Favorable Laws (LFLs) in Multi-Stream Non-Stationary Processes]
\label{def1 B}
For a fixed $\theta \in \Theta$, we say that the pair of families ($\{\mathcal{P}^{F, \theta}_{n} \}, \{\mathcal{P}^{G, \theta}_{n, \nu}\}) $ is stochastically bounded by the pair of least favorable laws (LFLs) $(\bar{F}^\theta, \bar{G}^\theta) = (\{\bar{f}_{\theta, n}\}, \{\bar{g}_{\theta, n, \nu}\})$, with $\bar{f}_{\theta, n} \in \mathcal{P}^{F, \theta}_{n}$ and 
$
\bar{g}_{\theta, n, \nu} \in \mathcal{P}^{G, \theta}_{n, \nu}, \; n, \nu=1,2, \dots,
$
or that the pair $(\bar{F}^\theta, \bar{G}^\theta)$ is least favorable if 
\begin{equation}\label{eq:stocbounded-1 multi-stream}
		\begin{split}
			 \bar f_{\theta,n}
            &\succ  {f}_{\theta,n}, \quad \text{ for every } \; f_{\theta,n} \in \mathcal{P}^{F, \theta}_{n},\quad n = 1,2, \dots, 
		\end{split}
	\end{equation}
and 
\begin{equation}\label{eq:stocbounded-2 multi-stream}
        \begin{split}
			 g_{\theta,n, \nu}
            &\succ 
			\bar{g}_{\theta, n, \nu}, \quad \text{ for every } \; g_{\theta,n, \nu} \in \mathcal{P}^{G, \theta}_{n, \nu}, \quad n,\nu=1,2, \dots.       
		\end{split}
    \end{equation}
We also say that the pair $(\bar{F}, \bar{G})$ are a pair of LFLs, where
$$
\bar{F} = \{\bar{F}^\theta\}_{\theta \in \Theta}, \quad \text{ and } \quad \bar{G} = \{\bar{G}^\theta\}_{\theta \in \Theta},
$$
if the pair $(\bar{F}^\theta, \bar{G}^\theta)$ is least favorable for each $\theta \in \Theta$. 
\end{definition}

Given a pair of LFLs $(\bar{F}, \bar{G})$, we use them to define our robust multi-stream detection algorithm. This is a cumulative sum algorithm based on generalized likelihood ratio statistic (computed using LFLs) and is given by
\begin{align}\label{G_n_B}
    \bar\Psi_{n} = \max_{1 \leq k \leq n} \max_{{B} \in \mathcal{B}} \sum_{\theta \in {B}} \sum_{i = k}^{n} \log\frac{\bar g_{\theta, i, k}(X_{\theta, i})}{\bar{f}_{\theta, i}(X_{\theta,i})}.
\end{align}
The corresponding stopping rule is
\begin{align}\label{eqn: tau glr B}
    \tau_{ms} = \inf \{n \geq 1 : \bar\Psi_{n} \geq {A_{n, \alpha}}\}. 
\end{align}	
Thresholds $ A_{n,\alpha}$ are chosen to satisfy the constraint on the mean time to a false alarm.

Finally, consider the case $K = 1$, where change happens in at most one stream. Then, $\mb = \{\{\theta_1\}, \dots, \{\theta_M\}\}$.  Therefore, the robust algorithm becomes the following:
\begin{align}\label{G_n_t}
    \bar\Phi_n = \max_{1 \leq k \leq n} \max_{\ta \in \Theta} \sum_{i = k}^{n} \log\frac{\bar g_{\theta, i, k}(X_{\theta, i})}{\bar{f}_{\theta, i}(X_{\theta,i})},
\end{align}
with a corresponding stopping rule
\begin{align}\label{eqn: tau glr t}
    \tau_{mc} = \inf \{n \geq 1 : \bar\Phi_n \geq {A_{n, \alpha}}\}. 
\end{align}	
In Section~\ref{sec:designMultiStream} below, we briefly discuss some standard sufficient conditions needed for the optimality of $\tau_{ms}$ (and hence of $\tau_{mc}$) when the true pre- and post-change laws are given by the LFLs. 

\subsection{Robust Optimality in a Multi-Stream Non-Stationary Setting}

The following Theorem~\ref{thm:2} is a multi-stream generalization of Theorem \ref{thm:LFLRobust_minimax}. 

\begin{theorem}\label{thm:2}
Suppose the following conditions hold:
\begin{enumerate}
    	\item[(a)] 	
     The pair of families $\{\mpp_n^{\fb}\}_{\ta\in\Theta}$ and $\{\mpp^{\gb}_{n, \nu}\}_{\ta\in\Theta}$ is stochastically bounded by the LFLs $\bar{F}$ and $\bar{G}$, where $
\bar{F} = \{\bar{F}^\theta\}_{\theta \in \Theta}$, and $\bar{G} = \{\bar{G}^\theta\}_{\theta \in \Theta}$. 
	\item[(b)] Let $\alpha \in (0,1)$ be a constraint on the rate of false alarms. Let the thresholds $\{A_{n, \alpha}\}$ used by the stopping rule $\tau_{ms}$ in \eqref{eqn: tau glr B} designed using the LFLs be such that 
    \begin{align}
    \label{eqn:E_infty_glr multi 2}
        \Expect_{\infty}^{{\bar{F}}}[\tau_{ms}] = \frac{1}{\alpha}.
    \end{align}
	\item[(c)] All likelihood ratio functions used in the stopping rule $\tau_{ms}$ are continuous and monotone increasing.
  \end{enumerate}
Then the following results are true:
    \begin{enumerate}
    \item If the stopping rule $\tau_{ms}$ in \eqref{eqn: tau glr B} is exactly optimal to \eqref{eq:QCDproblem-B}, uniformly over $B \in \mathcal{B}$, when the law is LFL, i.e., when $\mf = \{\bar{F}\}$ and  $\mg = \{\bar{G}\}$ as defined in \eqref{eq: dist family}, then the stopping rule is exactly robust optimal for the problem in \eqref{eqn:robustprob2} uniformly over $B \in \mathcal{B}$. 
    \item If the stopping rule $\tau_{ms}$ in \eqref{eqn: tau glr B} is asymptotically optimal to \eqref{eq:QCDproblem-B}, uniformly over $B \in \mathcal{B}$,  when the law is LFL, i.e., when $\mf = \{\bar{F}\}$ and  $\mg = \{\bar{G}\}$ as defined in \eqref{eq: dist family}, then the stopping rule is asymptotically robust optimal for the problem in \eqref{eqn:robustprob2}, uniformly over $B \in \mathcal{B}$, as $\alpha \to 0$. 
    \end{enumerate}
\end{theorem}

\begin{proof}
See Appendix \ref{appen:2}.
\end{proof}

\subsection{On Optimality of the Multi-Stream Robust Algorithm under LFLs}
\label{sec:designMultiStream}
In this section, we provide sufficient conditions under which the multi-stream robust algorithm $\tau_{ms}$ is optimal for the LFLs. The conditions are similar (but not identical) to the ones stated for the single-stream case in Theorem~\ref{thm:modifiedconds}.

Define
$$
\bar Z_{\theta, n, \nu} = \log \frac{\bar{g}_{\theta, n,\nu}(X_{\theta, n})}{\bar f_{\theta, n}(X_{\theta, n})}
$$
to be the log-likelihood ratio at time $n$, in stream $\theta$, and calculated for the LFLs when the change occurs at $\nu$. Let $I_\theta > 0$ be the information number (similar to the number $I$ in Theorem~\ref{thm:modifiedconds}) for stream $\theta$, and let 
$$
I_B = \sum_{\theta \in B} I_\theta. 
$$

\begin{theorem}[\cite{brucks2023modeling, lai-ieeetit-1998, tart-niki-bass-2014}]
	\label{thm:modifiedconds_mult}
	\begin{enumerate}
		\item Let there exist a positive numbers $I_\theta$, $\theta \in \Theta$, such that the log-likelihood ratios $\{\bar Z_{\theta, n, \nu}\}$ satisfy the following  condition: for any $\delta > 0$, 
		\begin{equation}
			\label{eq:Znnu_LB multi-stream}
			\begin{split}
		\hspace{-1.1cm}		\lim_{n \to \infty} \; \sup_{\nu \geq 1} \; \esssup \mathsf{P}_\nu^{\bar F, \bar{G}, B} &\left(\max_{t \leq n} \sum_{i = \nu }^{\nu + t} \sum_{\theta \in B}\bar Z_{\theta, i, \nu} \geq I_B(1+\delta)n \; \bigg| \; X_{\theta, 1}, \dots, X_{\theta, \nu-1}, \theta \in \Theta\right) = 0.
			\end{split}
		\end{equation}
		Then, we have the universal lower bound as $\alpha \to 0$, 
		\begin{equation}
			\begin{split}
				\inf_{\tau} \textup{WADD}^{\bar F, \bar{G}, B}(\tau) \; \geq \; \frac{|\log \alpha|}{I_B} (1+o(1)). 
			\end{split}
		\end{equation}
		Here the minimum over $\tau$ is over those stopping times satisfying $\Expect_\infty^{\bar{F}}[\tau] \geq \frac{1}{\alpha}$. 
		\item The robust algorithm $\tau_{ms}$ in \eqref{eqn: tau glr B}, when we choose the threshold $\bar A_{n, \alpha} = \log \frac{|\mathcal{B}|}{\alpha}$,
		$$
		\tau_{ms} = \min\left\{n \geq 1: \max_{1 \leq k \leq n} \max_{B \in \mathcal{B}} \sum_{i=k}^n \sum_{\theta \in B} \bar Z_{\theta, i,k} \geq \log \frac{|\mathcal{B}|}{\alpha}\right\},
		$$
		satisfies
		$$
		\mathsf{E}_\infty^{\bar F}[\tau_{ms}] \geq \frac{1}{\alpha}.
		$$
		\item Furthermore, if the log-likelihood ratios $\{\bar Z_{\theta, n, \nu}\}$ also satisfy, $\forall \delta > 0$, 
		\begin{equation}
			\label{eq:Znnu_UB multi-stream}
			\begin{split}
				\lim_{n \to \infty} \; \sup_{k \geq \nu \geq 1} \; \esssup \mathsf{P}_\nu^{\bar F, \bar{G}, B} & \left(\frac{1}{n}\sum_{i = k }^{k +n} \sum_{\theta \in \Theta}\bar Z_{\theta, i,k} \leq I_B - \delta \; \bigg| \; X_{\theta,1}, \dots, X_{\theta, k-1}, \theta \in \Theta\right) = 0.
			\end{split}
		\end{equation}
		Then as $\alpha \to 0$, $\tau_{ms}$ achieves the lower bound:
		\begin{equation}
			\begin{split}
				\textup{WADD}^{\bar F, \bar{G}, B}(\tau_{ms})  \leq  \frac{|\log \alpha |}{I_B}(1+o(1)), \quad \alpha \to 0. 
			\end{split}
		\end{equation}
	\end{enumerate}
\end{theorem}

All the examples given in Section~\ref{sec:designSingleStream} will work for the multi-stream case as well. For further examples, we refer to \cite{tart-niki-bass-2014, tart-book-2019}. 

\subsection{On Performance of the Multi-Stream Robust Algorithm under Laws other than LFLs}
The following result obtains bounds on the mean time to a false alarm and detection delay of the stopping rule $\tau_{ms}$ when the pre- and post-change laws are not necessarily the pair of LFLs. 
\begin{corollary}
    Suppose the log-likelihood ratios $\{\bar Z_{\theta, n, \nu}\}$ satisfy the conditions in Theorem~\ref{thm:modifiedconds_mult}. Then the multi-stream robust algorithm $\tau_{ms}$ in \eqref{eqn: tau glr B}, when we choose the threshold $\bar A_{n, \alpha} = \log \frac{|\mathcal{B}|}{\alpha}$,
		$$
		\tau_{ms} = \min\left\{n \geq 1: \max_{1 \leq k \leq n} \max_{B \in \mathcal{B}} \sum_{i=k}^n \sum_{\theta \in B} \bar Z_{\theta, i,k} \geq \log \frac{|\mathcal{B}|}{\alpha}\right\},
		$$
		satisfies for all $B \in \mathcal{B}$, 
        \begin{equation}
			\begin{split}
            \mathsf{E}_\infty^{ F}[\tau_{ms}] &\geq \frac{1}{\alpha}, \quad \forall F \in \mathcal{F}. \\
				\textup{WADD}^{ F, {G}, B}(\tau_{ms})  &\leq  \frac{|\log \alpha |}{I_B}(1+o(1)), \quad \alpha \to 0, \quad \forall F \in \mathcal{F}, G \in \mathcal{G}. 
			\end{split}
		\end{equation}
\end{corollary}

\section{Numerical Results}
\label{sec:NumericalResults}
In this section, we show the effectiveness of the robust tests on simulated and real data. In Section~\ref{sec:numerical_single}, we conduct experiments using the single-stream robust test:
\begin{equation} \label{eq:singlestream_robust}
\begin{split}
     \bar{W}_{n} &= \max_{1 \leq k \leq n} \sum_{i=k}^n \log \frac{\bar{g}_{i,k}(X_{i})}{\bar f_{i}(X_{i})}\\
    \tau_{ss} &= \inf\{n \geq 1 : \bar W_{n} \geq  A\}.
\end{split}
\end{equation}
In Section~\ref{sec:numerical_mult}, we give results on change detection using the multi-stream robust test:
\begin{equation}\label{eq:multstream_robust}
\begin{split}
    \bar\Psi_{n} = \max_{1 \leq k \leq n} \max_{{B} \in \mathcal{B}} \sum_{\theta \in {B}} \sum_{i = k}^{n} &\log\frac{\bar g_{\theta, i, k}(X_{\theta, i})}{\bar{f}_{\theta, i}(X_{\theta,i})}=\max_{{B} \in \mathcal{B}} \max_{1 \leq k \leq n}  \sum_{\theta \in {B}} \sum_{i = k}^{n} \log\frac{\bar g_{\theta, i, k}(X_{\theta, i})}{\bar{f}_{\theta, i}(X_{\theta,i})}\\
    \tau_{ms} &= \inf \{n \geq 1 : \bar\Psi_{n} \geq A\}. 
    \end{split}
\end{equation}	
When only one stream is affected post-change, then $\tau_{ms}$ simplifies to
\begin{equation}\label{eq:multstream_robust_onechange}
\begin{split}
    \bar\Phi_n &= \max_{1 \leq k \leq n} \max_{\ta \in \Theta} \sum_{i = k}^{n} \log\frac{\bar g_{\theta, i, k}(X_{\theta, i})}{\bar{f}_{\theta, i}(X_{\theta,i})}, \\
    \tau_{mc} &= \inf \{n \geq 1 : \bar\Phi_n \geq A\}. 
    \end{split}
\end{equation}		
We keep the thresholds constant over time for ease of implementation and due to the asymptotic optimality of single-threshold versions (see Theorems~\ref{thm:modifiedconds} and Theorem~\ref{thm:modifiedconds_mult}). For the application of the multi-stream algorithms, we also report the identity of the stream(s) that caused the change using 
\begin{equation}\label{eq:argmax_to_identify_stream}
\begin{split}
    \hat B &= \arg \max_{{B} \in \mathcal{B}} \max_{1 \leq k \leq n}  \sum_{\theta \in {B}} \sum_{i = k}^{n} \log\frac{\bar g_{\theta, i, k}(X_{\theta, i})}{\bar{f}_{\theta, i}(X_{\theta,i})}, \; \text{ or } \\
    \hat \theta &= \arg \max_{\ta \in \Theta} 
 \max_{1 \leq k \leq n} \sum_{i = k}^{n} \log\frac{\bar g_{\theta, i, k}(X_{\theta, i})}{\bar{f}_{\theta, i}(X_{\theta,i})}. 
    \end{split}
\end{equation}

\subsection{Numerical Results for Single-Stream Data}
\label{sec:numerical_single}

\subsubsection{Application to Simulated Gaussian and Poisson data}
For the Gaussian case, we choose the real-time data process with the following model:
\begin{equation}\label{eq: real data normal}
    \begin{split}
        f_n = \mathcal{N}(\theta_n,1), \quad  g_{n, \nu} = \mathcal{N}(\mu_{n,\nu}, 1), \quad \theta_n \in [0, 1], \quad \mu_{n,\nu} \in [2, 3].
    \end{split}
\end{equation}
Following the result in Example \ref{exam:GaussLFD}, pre-LFL is $\mathcal{N}(1, 1)$, and post-LFL is $\mathcal{N}(2, 1)$. Thus, single-stream robust test in \eqref{eq:singlestream_robust} designed with 
$$\bar f_n =\bar{f}= \mathcal{N}(1, 1) \quad  \text{ and } \quad \bar g_{n, \nu} = \bar{g} = \mathcal{N}(2, 1)
$$ 
is robust optimal. 

We generate the observations $\{X_n\}$ in the following three ways:
\begin{itemize}
\item [1)] LFL observations: $\{X_n\}$ follow LFL $\bar f_n = \bar f = \mathcal{N}(1, 1)$ before the change and $\bar g_{n, \nu} = \bar g = \mathcal{N}(2, 1)$ after the change;
\item [2)] Random observations: $\{X_n\}$ follow $f_n = \mathcal{N}(\mu_{0, n}, 1)$ and $g_{n, \nu} = g_n = \mathcal{N}(\mu_{1, n}, 1)$, with $\mu_{0, n} \sim \text{Unif}[0, 1]$ and $\mu_{1, n} \sim \text{Unif}[2, 3]$, i.e., mean values are randomly generated in their respective ranges;
\item [3)] I.P.I.D. (independent and periodically identically distributed) observations: $\{X_n\}$ follow $f_n = \mathcal{N}(\mu_{0, n}, 1)$ and $g_{n, \nu} = \mathcal{N}(\mu_{1, n}, 1)$, with $\mu_{0,n}$ selected periodically with period $11$ from the set $\{0, 0.1, \dots, 1\}$ (in that order), and $\mu_{1,n}$ selected periodically with period $11$ from the set $\{2, 2.1, \dots, 3\}$. 
\end{itemize}

Figure \ref{fig:robust cusum} demonstrates the effectiveness of $\tau_{ss}$, compared to non-robust tests (CUSUM with randomly picked $f$ and $g$). 
The threshold is set to $5 = \log(150)$ to satisfy average run length to false alarm $\geq 150$.
In Figure \ref{fig:robust cusum} [Left], change time $(\nu = 23)$ is detected by $\tau_{ss}$ in the three scenarios of $\{X_n\}$ observations.
In Figure \ref{fig:robust cusum} [Middle], robust test ($\tau_{ss}$) uses 10 different random observation $\{X_n\}$ $(\nu = 54)$, with detection delays less than 10.
In Figure \ref{fig:robust cusum} [Right], the non-robust test is CUSUM \eqref{eq:iidcusum} with $f\sim\mathcal{N}(0, 1)$ and $g\sim\mathcal{N}(3, 1)$. It uses 10 different LFL observations $(\nu = 23)$. Even though the threshold is raised to be $6.9 = \log(1000)$ to prevent false alarms, premature stopping remains in 2 out of 10 cases.
\begin{figure}[!]
    \centering
    \includegraphics[scale=0.355]{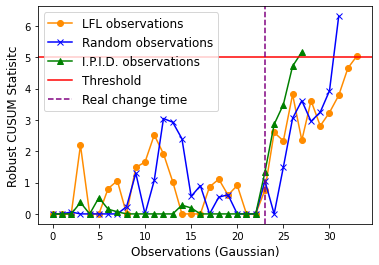}
    \includegraphics[scale=0.355]{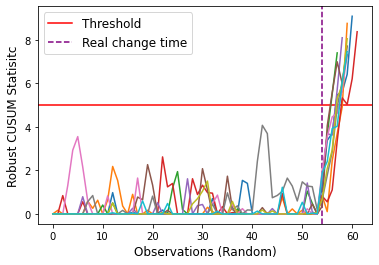}
    \includegraphics[scale=0.355]{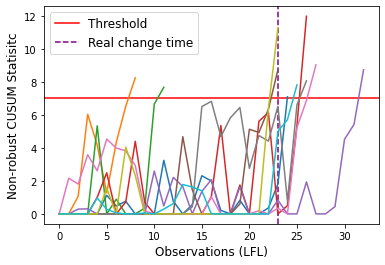}
    \caption{{Evaluating effectiveness of $\tau_{ss}$ with Gaussian data following \eqref{eq: real data normal}.} 
    [Left] $\tau_{ss}$ with $\bar f_n =\mathcal{N}(1,1)$ and $\bar g_{n, \nu} = \mathcal{N}(2, 1)$ tested on three scenarios of $\{X_n\}$. 
    [Middle] $\tau_{ss}$ tested on 10 random observations $\{X_n\}$. 
    [Right] Non-robust CUSUM test with $f = \mathcal{N}(0,1)$ and $g = \mathcal{N}(3, 1)$, tested on 10 LFL observation $\{X_n\}$. }
    \label{fig:robust cusum}
\end{figure}


Similarly, for the Poisson case, we choose the following model:
\begin{equation}\label{eq: real data poisson}
    \begin{split}
        f_n = \text{Pois}(\gamma_n), \quad  g_{n, \nu} = \text{Pois}(\lambda_{n,\nu}), \quad  \gamma_n \in [0.4, 0.5], \quad \lambda_{n,\nu} \in [1, 1.1].
    \end{split}
\end{equation}
Following Example \ref{exam:PoissonLFD}, $\tau_{ss}$ in \eqref{eq:singlestream_robust} designed with pre-LFL $\bar f_n = \text{Pois}(0.5)$ and post-LFL $\bar g_{n, \nu} = \text{Pois}(1)$ is robust optimal. 
Similar to the Gaussian case above, we consider three scenarios of $\{X_n\}$ observations:
\begin{itemize}
\item [1)] LFL observations: $\{X_n\}$ follow LFL $\bar f_n = \bar f = \text{Pois}(0.5)$ before the change and $\bar g_{n, \nu} = \bar g = \text{Pois}(1)$ after the change;
\item [2)] Random observations: $\{X_n\}$ follow $f_n = \text{Pois}(\lambda_{0, n})$ and $g_{n, \nu} = g_n = \text{Pois}(\lambda_{1, n}, 1)$, with $\lambda_{0, n} \sim \text{Unif}[0.4, 0.5]$ and $\lambda_{1, n} \sim \text{Unif}[1, 1.1]$, i.e., the parameter values are randomly generated in their respective ranges;
\item [3)] I.P.I.D. (independent and periodically identically distributed) observations: $\{X_n\}$ follow $f_n = \text{Pois}(\lambda_{0, n}, 1)$ and $g_{n, \nu} = \text{Pois}(\lambda_{1, n}, 1)$, with $\lambda_{0,n}$ selected periodically with period $11$ from the set $\{0.4, 0.41, \dots, 0.5\}$ (in that order), and $\lambda_{1,n}$ selected periodically with period $11$ from the set $\{1, 1.01, \dots, 1.1\}$.
\end{itemize}

\begin{figure}[t]
    \centering
    \includegraphics[scale=0.355]{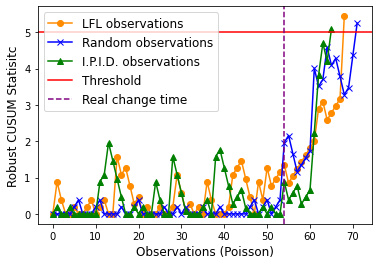}
    \includegraphics[scale=0.355]{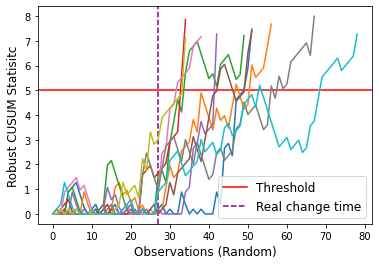}
    \includegraphics[scale=0.355]{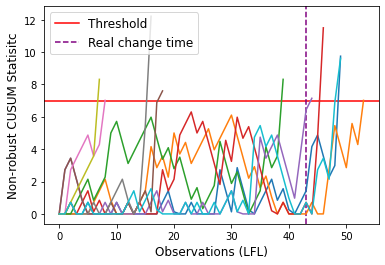}
    \caption{Evaluating effectiveness of $\tau_{ss}$ with Poisson data following \eqref{eq: real data poisson}.
    [Left] $\tau_{ss}$ with $\bar f_n =\text{Pois}(0.5)$ and $\bar g_{n, \nu} = \text{Pois}(1)$ tested on three scenarios of $\{X_n\}$. 
    [Middle] $\tau_{ss}$ tested on 10 random observations $\{X_n\}$. 
    [Right] Non-robust CUSUM test with $f = \text{Pois}(0.4)$ and $g = \text{Pois}(1.1)$, tested on 10 LFL observation $\{X_n\}$. }
    \label{fig:robust cusum poisson}
\end{figure}

Figure~\ref{fig:robust cusum poisson} demonstrated the effectiveness of $\tau_{ss}$. The threshold is set to $5 = \log(150)$ to satisfy average run length to false alarm $\geq 150$.
In Figure \ref{fig:robust cusum poisson} [Left], change $(\nu = 54)$ is detected for different observation cases.
In Figure \ref{fig:robust cusum poisson} [Middle], $\tau_{ss}$ detects the change $(\nu = 27)$ with LFL $\{X_n\}$.
In Figure \ref{fig:robust cusum poisson} [Right], the non-robust test is CUSUM \eqref{eq:iidcusum} with $f\sim\text{Pois}(0.4)$ and $g\sim\text{Pois}(1.1)$. It uses 10 different LFL observations $(\nu = 43)$. Even though the threshold is raised to be $6.9 = \log(1000)$ to prevent false alarms, such premature stopping remains in 5 out of 10 cases.


\subsubsection{Application to Detecting Multiple Waves of COVID-19}
In this section, we report results on the application of the robust test $\tau_{ss}$ in \eqref{eq:singlestream_robust} to detect the onset and second outbreak of a pandemic, via the publicly available U.S. COVID-19 infection dataset. We selected two U.S. counties of similar populations: Allegheny, Pennsylvania, and St. Louis, Missouri. The daily number of cases for the first $200$ days (starting from 2020/1/22) are plotted in Fig.~\ref{fig:data_stat_Allegheny} (Left) and Fig.~\ref{fig:data_stat_StLouis} (Left), with $\text{Pois}(1)$ noise added to the data. 
This data generation process simulates the scenario where it is required to detect the onset (or second outbreak) of a pandemic in the backdrop of daily infections due to other viruses or to detect the arrival of a new variant. 
Thus, the data represents the detection of deviation from a baseline. 
In Figure ~\ref{fig:data_stat_Allegheny} (Left) and Figure~\ref{fig:data_stat_StLouis} (left), the infection rates started to become significant (in both counties) around day 50. The second outbreak is around day 150 in Allegheny and day 140 in St. Louis.

\begin{figure}[ht]
    \centering
    \includegraphics[scale=0.355]{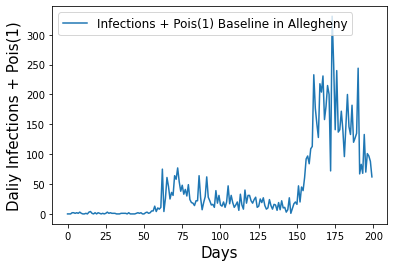}
    \includegraphics[scale=0.355]{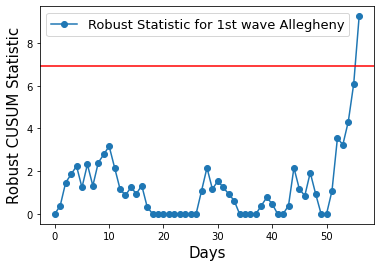}
    \includegraphics[scale=0.355]{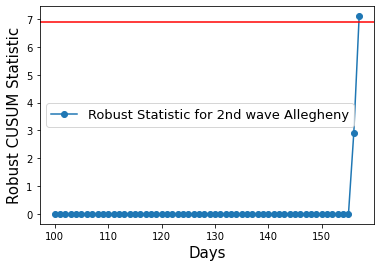}
    \caption{Robust test $\tau_{ss}$ for detecting COVID-19 outbreak in Allegheny County. 
    [Left] Daily increase of confirmed infection cases with $\text{Pois}(1)$ noise added. 
    [Middle] For the first-wave detection, $\tau_{ss}$ is designed with pre-LFL being $\text{Pois}(1)$ and post-LFL being $\text{Pois}(2)$. 
    [Right] For the second-wave detection starting from day 100, $\tau_{ss}$ is designed with pre-LFL being $\text{Pois}(70)$ and post-LFL being $\text{Pois}(93)$. 
    }
    \label{fig:data_stat_Allegheny}
\end{figure}

\begin{figure}[ht]
    \centering
    \includegraphics[scale=0.355]{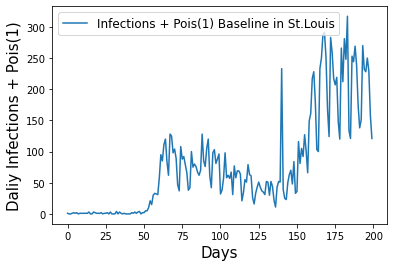}
    \includegraphics[scale=0.355]{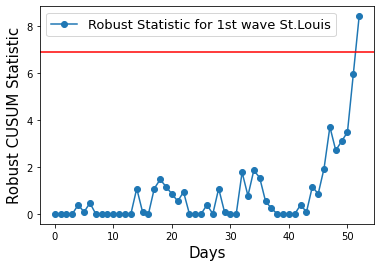}
    \includegraphics[scale=0.355]{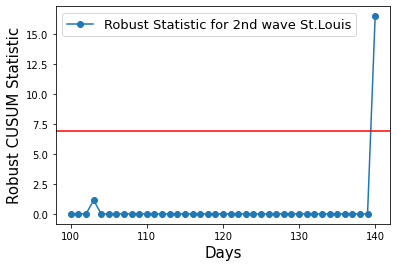}
    \caption{Robust test $\tau_{ss}$ for detecting COVID-19 outbreak in St. Louis County. 
    [Left] Daily increase of confirmed infection cases with $\text{Pois}(1)$ noise added. 
    [Middle] For the first-wave detection, $\tau_{ss}$ is designed with pre-LFL being $\text{Pois}(1)$ and post-LFL being $\text{Pois}(2)$. 
    [Right] For the second-wave detection starting from day 100, $\tau_{ss}$ is designed with pre-LFL being $\text{Pois}(138)$ and post-LFL begin $\text{Pois}(171)$. 
    }
    \label{fig:data_stat_StLouis}
\end{figure}

First, we detect the onset of COVID-19. Since the actual number of infections is unknown and can change over time, we take the robust approach and design the robust test with pre-LFL being $\text{Pois}(1)$, and post-LFL being $\text{Pois}(2)$. The threshold is chosen to be $6.9 = \log(1000)$ to guarantee a mean-time to false alarm greater than $1000$ days.  
As seen in Figure~\ref{fig:data_stat_Allegheny} (Middle) and Figure~\ref{fig:data_stat_StLouis} (Middle), the robust CUSUM test detects the change quickly (within a week). 

For the second wave detection, we start at day 100, when the daily infections become stable.
To estimate the LFL, we use daily infections from the detected day of the first wave (day 56 for Allegheny and 52 for St.Louis) to day 99 (before second-wave detection).
For the pre-LFL, we choose $\text{Pois}(\mu + 2\sigma)$, with $\mu$ and $\sigma$ estimated by sample mean and standard deviation.
The resulting pre-LFL is $\text{Pois}(70)$ for Allegheny, and $\text{Pois}(138)$ for St. Louis.
For the post-LFL, to align with condition \eqref{eq:stocbounded-2}, we choose $\text{Pois}(\mu + 3\sigma)$, resulting $\text{Pois}(93)$ for Allegheny, and $\text{Pois}(171)$ for St. Louis.
In Figure~\ref{fig:data_stat_Allegheny} (Right) and Figure~\ref{fig:data_stat_StLouis} (Right), the test $\tau_{ss}$ detects the second outbreak quickly (day 157 and day 140), around the days when infections start to become significant again. Note that in St. Louis, day 140 has an obvious high infection, and it is not due to the noise added. Our algorithm quickly detects such change without any delay.

\subsection{Numerical Results for Multi-Stream Data}
\label{sec:numerical_mult}

\subsubsection{Application to Simulated Data with Change in One Stream}
\label{sec: 6.2.1}

We discuss the performance of multi-stream robust algorithm $\tau_{mc}$ in \eqref{eq:multstream_robust_onechange}, where change is assumed to happen in one stream. 
For the Gaussian case, we choose the following model:
\begin{equation}
\label{eq: multi simulation 2}
	\begin{split}
		f_{\theta} = \mathcal{N}(1, 1), \quad
		g_{\theta, n, \nu} = \mathcal{N}(\mu_{\theta, n, \nu}, 1), \quad
		\mu_{\theta, n, \nu} \geq 1.5,\quad
		\theta = 0, 1, 2.
	\end{split}
\end{equation}
It follows that the post-change LFL for each stream is
$$
\bar g_{\theta, n, \nu} = \mathcal{N}(1.5, 1), \quad \forall n, \nu, \theta. 
$$
We set the change time $\nu = 10$, and only stream $\theta = 1$ changes to $\mathcal{N}(\mu_{\theta, n, \nu}, 1)$. We consider the following three experiments:
\begin{itemize}
    \item [1)] At the change point, the density of observations in the stream indexed by $\theta=1$ changes to a stationary distribution $g_{\theta, n, \nu} = \mathcal{N}(\mu_{\theta, n, \nu}, 1)$ with $\mu_{\theta, n, \nu} = 2$ while the density of observations in streams $\theta=0$ and $\theta=2$ remains $\mathcal{N}(1, 1)$;
    \item [2)] At the change point, the density of observations in the stream indexed by $\theta=1$ changes to a stationary distribution $g_{\theta, n, \nu} = \mathcal{N}(\mu_{\theta, n, \nu}, 1)$ with $\mu_{\theta, n, \nu} = 4$ while the density of observations in streams $\theta=0$ and $\theta=2$ remains $\mathcal{N}(1, 1)$;
    \item [3)] At the change point, the density of observations in the stream indexed by $\theta=1$ changes to a non-stationary distribution with $g_{\theta, n, \nu} = \mathcal{N}(\mu_{\theta, n, \nu}, 1)$, $\mu_{\theta, n, \nu}\sim\text{Unif}[2, 4]$ and the mean value is randomly chosen for each $n$.
\end{itemize}

\begin{figure}[ht]
	\centering
	\includegraphics[scale=0.5]{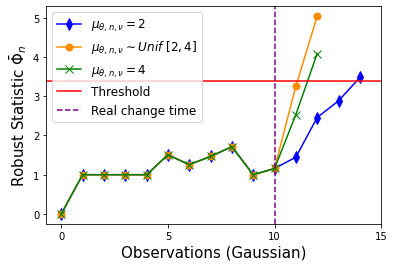}
	\includegraphics[scale=0.5]{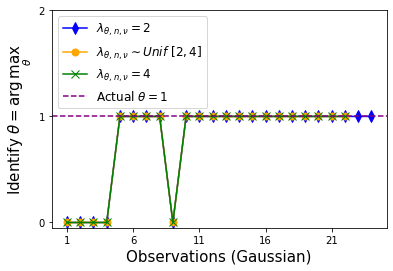}
    \caption{[Left] Detecting the change in the three experiments using robust algorithm designed using $\bar f_{\theta, n} = \mathcal{N}(1, 1)$ and $\bar g_{\theta, n, \nu} = \mathcal{N}(1.5, 1)$ for each $\theta$. 
    [Right] Identifying post-change scenario  by $\arg\max_{\ta\in\Theta}$, with the actual change in $\theta = 1$. }
    \label{fig:single normal}
\end{figure}

With $|\Theta| = 3$, the robust statisitc $\bar\Phi_n$ \eqref{eq:multstream_robust_onechange} with LFL is calculated, with threshold being $\log(|\Theta|/\alpha) = \log(10|\Theta|)$.
Figure \ref{fig:single normal} [Left] shows that the detection delays for the three experiments are 4, 2, and 2, respectively.
Figure \ref{fig:single normal} [Right] identifies the change in $\theta = 1$, induced by $\arg\max_{\ta\in\Theta}$ in $\bar\Phi_n$, with 10 additional observations taken after detected change time to ensure stable conclusion. Namely, when we stop by the detected change time, $\ta$ is decided by $\arg\max_{\ta\in\Theta}$ in \eqref{eq:argmax_to_identify_stream}. The identification aligns with the true post-change situation. 

Note that identifying the exact post-change scenario involves the data isolation process, and there is a possibility for false identification (\cite{warner2024worst}). The exhibition of the detected post-change scenario is just to show the possible validity of our mutli-stream robust algorithm. However, we can apply the single stream robust algorithm to each stream separately, and the first stream hitting the threshold is the actual post-change scenario.

For the Poisson case, we choose the following model:
\begin{equation}
\label{eq: multi simulation 1}
	\begin{split}
		f_{\theta, n} = \text{Pois}(1), \quad
		g_{\theta, n, \nu} = \text{Pois}(\lambda_{\theta, n, \nu}), \quad
		\lambda_{\theta, n, \nu} \geq 1.5,\quad
		\theta = 0, 1, 2.
	\end{split}
\end{equation}
It follows that the post-change LFL for each stream is
$$
\bar g_{\theta, n, \nu} = \text{Pois}(1.5), \quad \forall n, \nu, \theta. 
$$
Similar to the Gaussian case, we set the change time $\nu = 10$ and consider the following three experiments: 
\begin{itemize}
    \item [1)] At the change point, the density of observations in the stream indexed by $\theta=1$ changes to a stationary distribution $g_{\theta, n, \nu} = \text{Pois}(\lambda_{\theta, n, \nu} = 2)$ while the density of observations in streams $\theta=0$ and $\theta=2$ remains \text{Pois}(1);
    \item [2)] At the change point, the density of observations in the stream indexed by $\theta=1$ changes to a stationary distribution $g_{\theta, n, \nu} = \text{Pois}(\lambda_{\theta, n, \nu} = 4)$ while the density of observations in streams $\theta=0$ and $\theta=2$ remains \text{Pois}(1);
    \item [3)] At the change point, the density of observations in the stream indexed by $\theta=1$ changes to a non-stationary distribution with $g_{\theta, n, \nu} = \text{Pois}(\lambda_{\theta, n, \nu})$, $\lambda_{\theta, n, \nu}\sim\text{Unif}[2, 4]$ and the parameter value is randomly chosen for each $n$.
\end{itemize}

\begin{figure}[ht]
	\centering
	\includegraphics[scale=0.5]{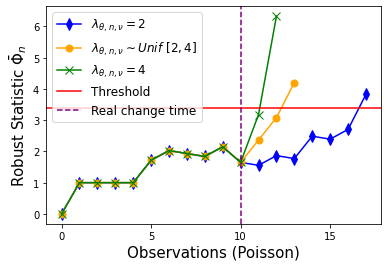}
	\includegraphics[scale=0.5]{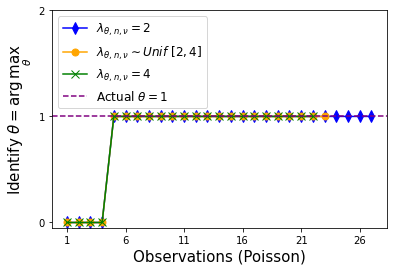}
    \caption{[Left] Detecting the change in three experiments using the robust algorithm designed using $\bar f_{\theta, n} = \text{Pois}(1)$ and $\bar g_{\theta, n, \nu} = \text{Pois}(1.5)$ for all $\theta$. [Right] Identifying post-change scenario by $\arg\max_{\ta\in\Theta}$, with the actual change in $\theta = 1$. }
    \label{fig:single poisson}
\end{figure}

With $|\Theta| = 3$, the robust test statistic $\bar\Phi_n$ with LFL is calculated with threshold $\log(|\Theta|/\alpha) = \log(10|\Theta|)$. 
Figure \ref{fig:single poisson} [Left] shows that the detection delays for the three cases are 7, 2, and 3, respectively.
Figure \ref{fig:single poisson} [Right] identifies the change in $\theta = 1$, induced by $\arg\max_{\ta\in\Theta}$ in \eqref{eq:argmax_to_identify_stream}, with 10 additional observations taken after detected change time to ensure stable conclusion. The identification aligns with the true post-change situation.

\subsubsection{Application to Simulated Data with Change in More Than One Stream}

We discuss the performance of multi-stream robust algorithm $\tau_{ms}$ in \eqref{eq:multstream_robust}, where change is assumed to happen in at most $K$ streams, and the post-change streams are indexed by $\ta\in B$. 
For the Gaussian case, we choose the following model:
\begin{equation}
\label{eq: multi simulation 4}
	\begin{split}
		f_{\theta, n} = \mathcal{N}(1, 1), \quad
		g_{\theta, n, \nu} = \mathcal{N}(\mu_{\theta, n, \nu}, 1), \quad
		\lambda_{\theta, n, \nu} \geq 1.5,\quad
		\theta = 0, 1, 2.
	\end{split}
\end{equation}
It follows that the post-change LFL for each stream is
$$
\bar g_{\theta, n, \nu} = \mathcal{N}(1.5, 1), \quad \forall n, \nu, \theta. 
$$
With $K = 2$, possible situations of change in subset of streams is
$$\mathcal{B} = \{\{0\}, \{1\}, \{2\}, \{0, 1\}, \{0, 2\}, \{1, 2\}\}.$$ 
We set the change time $\nu = 10$, and subset of streams indexed by $\theta \in B = \{0, 1\}$ change to $\mathcal{N}(\mu_{\theta, n, \nu}, 1)$. We consider the following three experiments:
\begin{itemize}
    \item [1)] At the change point, the density of observations in the streams indexed by $\theta=0, 1$ change to a stationary distribution $g_{\theta, n, \nu} = \mathcal{N}(\mu_{\theta, n, \nu}, 1)$ with $\mu_{\theta, n, \nu} = 2$, while the density of observations in stream $\theta=2$ remains $\mathcal{N}(1, 1)$;
    \item [2)] At the change point, the density of observations in the streams indexed by $\theta=0, 1$ change to a stationary distribution $g_{\theta, n, \nu} = \mathcal{N}(\mu_{\theta, n, \nu}, 1)$ with $\mu_{\theta, n, \nu} = 4$, while the density of observations in stream $\theta=2$ remains $\mathcal{N}(1, 1)$;
    \item [3)] At the change point, the density of observations in the streams indexed by $\theta=0, 1$ change to a non-stationary distribution $g_{\theta, n, \nu} = \mathcal{N}(\mu_{\theta, n, \nu}, 1)$ with $\mu_{\theta, n, \nu}\sim\text{Unif}[2, 4]$ and the mean value is randomly chosen for each $n$, while the density of observations in stream $\theta=2$ remains $\mathcal{N}(1, 1)$.
\end{itemize}

\begin{figure}[ht]
	\centering
	\includegraphics[scale=0.5]{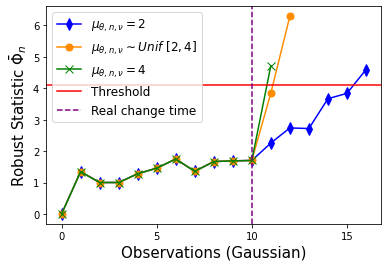}
	\includegraphics[scale=0.5]{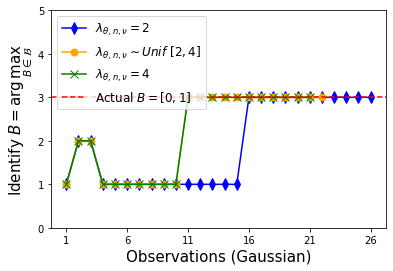}
    \caption{[Left] Detecting change with multi-stream robust statistic $\bar\Psi_n$ using $\bar f_{\theta, n} = \mathcal{N}(1, 1)$ and $\bar g_{\theta, n, \nu} = \mathcal{N}(1.5, 1)$. Three scenarios of $\{X_{\ta, n}\}_{\ta\in\Theta}$ are generated with post-change $\mathcal{N}(\mu_{\theta, n, \nu}, 1)$. [Right] Identifying post-change scenario by $\arg\max_{B \in \mathcal{B}}$, given that the actual change happens in $\theta\in B = \{0, 1\}$ (index 3 in $\mathcal{B}$). }
    \label{fig:multi normal}
\end{figure}

The multi-stream robust statistic $\bar\Psi_n$ in \eqref{eq:multstream_robust} is calculated, with the threshold being $\log(|\mathcal{B}|/\alpha) = \log(10|\mathcal{B}|)$. 
In Figure \ref{fig:multi normal} [Left], the detection delay for the three experiments are 6, 1, and 2, respectively. 
In Figure \ref{fig:multi normal} [Right], the change is detected in stream $\theta \in B = \{0, 1\}$ (index 3 in $\mathcal{B}$), which is consistent with the truth. Such identification of post-change scenario is induced by $\arg\max_{B \in \mathcal{B}}$ in \eqref{eq:argmax_to_identify_stream}, and 10 additional observations are taken after detected change time to ensure stable result. 

For the Poisson case, we choose the following model:
\begin{equation}
\label{eq: multi simulation 3}
	\begin{split}
		f_{\theta, n} = \text{Pois}(1), \quad
		g_{\theta, n, \nu} = \text{Pois}(\lambda_{\theta, n, \nu}), \quad
		\lambda_{\theta, n, \nu} \geq 1.5,\quad
		\theta = 0, 1, 2.
	\end{split}
\end{equation}
It follows that the post-change LFL for each stream is
$$
\bar g_{\theta, n, \nu} = \text{Pois}(1.5), \quad \forall n, \nu, \theta. 
$$
With $K = 2$, possible situations of change in subset of streams is
$$\mathcal{B} = \{\{0\}, \{1\}, \{2\}, \{0, 1\}, \{0, 2\}, \{1, 2\}\}.$$ 
We set the change time $\nu = 10$, and subset of streams indexed by $\theta \in B = \{1, 2\}$ change to $\mathcal{N}(\mu_{\theta, n, \nu}, 1)$.
We consider the following three experiments:
\begin{itemize}
    \item [1)] At the change point, the density of observations in the streams indexed by $\theta=1, 2$ change to a stationary distribution $g_{\theta, n, \nu} = \text{Pois}(\lambda_{\theta, n, \nu})$ with $\lambda_{\theta, n, \nu} = 2$, while the density of observations in stream $\theta=0$ remains $\text{Pois}(1)$;
    \item [2)] At the change point, the density of observations in the streams indexed by $\theta=1, 2$ change to a stationary distribution $g_{\theta, n, \nu} = \text{Pois}(\lambda_{\theta, n, \nu})$ with $\lambda_{\theta, n, \nu} = 4$, while the density of observations in stream $\theta=0$ remains $\text{Pois}(1)$;
    \item [3)] At the change point, the density of observations in the streams indexed by $\theta=1, 2$ change to a non-stationary distribution $g_{\theta, n, \nu} = \text{Pois}(\lambda_{\theta, n, \nu})$ with $\lambda_{\theta, n, \nu}\sim\text{Unif}[2, 4]$ and the mean value is randomly chosen for each $n$, while the density of observations in stream $\theta=0$ remains $\text{Pois}(1)$.
\end{itemize}

\begin{figure}[ht]
	\centering
	\includegraphics[scale=0.5]{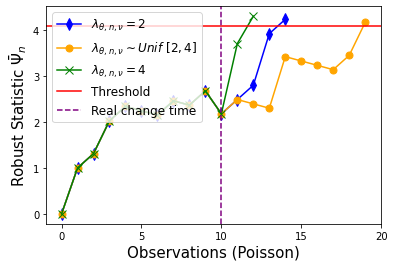}
	\includegraphics[scale=0.5]{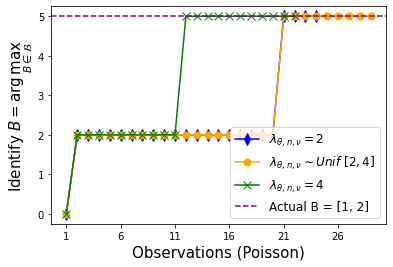}
    \caption{[Left] 
    Detecting change time with multi-stream robust statistic $\bar\Psi_n$ in \eqref{eq:multstream_robust} using $\bar f_{\theta, n} = \text{Pois}(1)$ and $\bar g_{\theta, n, \nu} = \text{Pois}(1.5)$. The three scenarios of observational data are generated following \eqref{eq: multi simulation 3}. [Right] Identifying post-change scenario by $\arg\max_{B \in \mathcal{B}}$, given that the actual change happens in $\theta\in B = \{1, 2\}$ (index 5 in $\mathcal{B}$). }
    \label{fig:multi poisson}
\end{figure}

The multi-stream robust statistic $\bar\Psi_n$ in \eqref{eq:multstream_robust} is calculated, with the threshold being $\log(|\mathcal{B}|/\alpha) = \log(10|\mathcal{B}|)$. 
In Figure \ref{fig:multi poisson} [Left], the detection delay for the three experiments are 4, 2, and 9, respectively.
In Figure \ref{fig:multi poisson} [Right], the changes is identified in stream $\theta \in B = \{1, 2\}$ (index 5 in $\mathcal{B}$), aligning with the truth. Such identification of post-change scenario is induced by $\arg\max_{B \in \mathcal{B}}$ in \eqref{eq:argmax_to_identify_stream}, and 10 additional observations are taken after detected change time to ensure stable result.

\subsubsection{Applying Multi-stream Robust Algorithm to Pittsburgh Flight Data}

In this section, we apply the multi-stream robust algorithm $\tau_{ms}$ in \eqref{eq:multstream_robust} to detect the arrival of aircraft based on data collected around Pittsburgh-Butler Regional Airport (\cite{Patrikar2021}). 
We assume that change can occur in at most $K = 3$ streams. 

The distance measurements for the last 100 seconds for randomly selected 35 flights are shown in Figure \ref{fig: flight distance}. We converted these distance measurements to signals by using the transformation $10/\text{Distance}$ (see Figure \ref{fig: flight distance} [Right]). $\mathcal{N}(0, 1)$ noise is added after padding 10 zeros at the beginning of the signals (see Figure \ref{fig: flight signal} [Right]). The noise is also added to simulate a scenario where an approaching enemy is being detected in a noisy environment.

\begin{figure}[ht]
	\centering
	\includegraphics[scale=0.5]{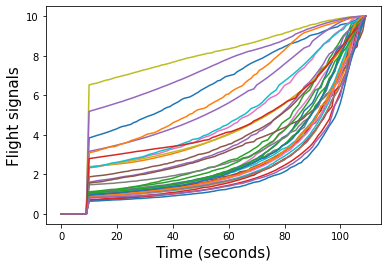}
	\includegraphics[scale=0.5]{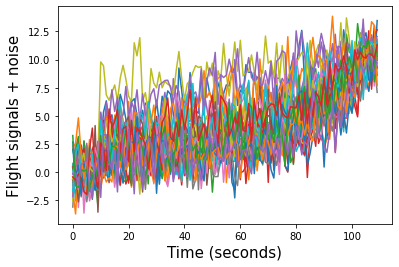}
    \caption{[Left] Flight signals for the last 100 seconds of 35 randomly chosen aircraft arriving at the Pittsburgh-Bulter Regional Airport (\cite{Patrikar2021}). Signals are padded with zeros, with the first 10 seconds as time before aircraft appear in the sensor system. [Right] Flight signal corrupted by $\mathcal{N}(0, 1)$ noise.}
    \label{fig: flight signal}
\end{figure}

\begin{figure}[ht!]
	\centering
	\includegraphics[scale=0.5]{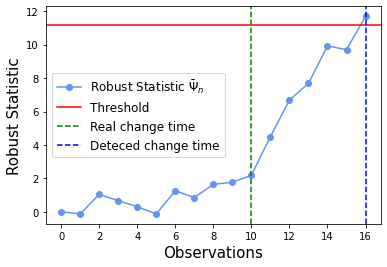}
	\includegraphics[scale=0.5]{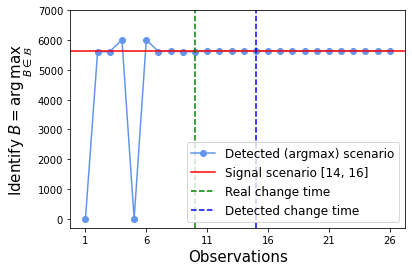}
    \caption{[Left] Detecting change via multi-stream robust statistic $\bar\Psi_n$ with $\bar f_{\theta, n} = \mathcal{N}(0, 1)$ and $\bar g_{\ta, n, \nu} = \mathcal{N}(0.5, 1)$. Threshold is $\log(|\mathcal{B}|/\alpha) = \log(7175 \times 10)$. 
    [Right] Identifying arrived flights, induced by $\arg\max_{B \in \mathcal{B}}$ from $\bar\Psi_n$, given that the true flights identities are $\{14, 16\}$.}
    \label{fig: flight signal detect}
\end{figure}

To simulate the asynchronous arrival of different flights, we randomly select 1 to 3 flights among the 35 flights to create the subset $B \subset [35]$. We then generate data streams as follows: we create $35$ streams of data, one for each flight. For streams not in $B$, the data remains $\mathcal{N}(0,1)$ before and after change. For streams in $B$, the data corresponds to the flight data for that stream corrupted by $\mathcal{N}(0,1)$ noise. 
Thus, the collection of possible post-change scenarios $\mathcal{B}$ is composed of the combination of $C^{35}_1$, $C^{35}_2$, and $C^{35}_3$, making $|\mathcal{B}| = 7175$. 
In Figure \ref{fig: flight signal detect}, we apply the multi-stream robust algorithm \eqref{eq:multstream_robust} to detect change time and identify the detected flights by \eqref{eq:argmax_to_identify_stream}, with 10 additional observations are taken after detected change time to ensure stable result. Specifically, we use the multi-stream robust statistic $\bar\Psi_n$ with $\bar f_{\theta, n} = \mathcal{N}(0, 1)$ and $\bar g_{\ta, n, \nu} = \mathcal{N}(0.5, 1)$. The threshold is set to $\log(|\mathcal{B}|/\alpha) = \log(7175 \times 10)$. 

\subsubsection{Applying Multi-stream Robust Algorithm to COVID Infection Data}
We apply the special case of multi-stream robust algorithm $\tau_{mc}$ in \eqref{eq:multstream_robust_onechange} to detect the earliest onset of a pandemic among counties in one state. The daily confirmed cases in Alabama (AL) and Pennsylvania (PA) for the first 150 days (starting from 2020/1/22) are plotted in Figure \ref{fig: covid infection}, and Pois(1) noise are added in Figure \ref{fig:AL infection} [Upper left] and Figure \ref{fig:PA infection} [Upper left]. 
County with the earliest infections (referring to the highest infections in early days) and the highest infections (for most of the time) are highlighted. We want to detect the county with the earliest infections, instead of the highest. 

\begin{figure}[!t]
	\centering
	\includegraphics[scale=0.45]{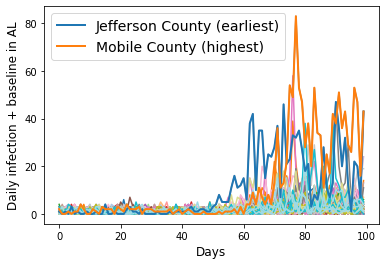}
	\includegraphics[scale=0.45]{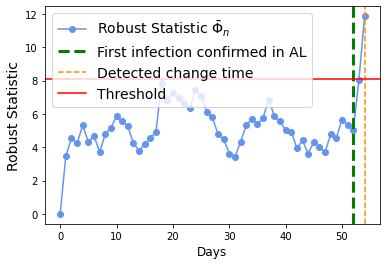}
	\includegraphics[scale=0.45]{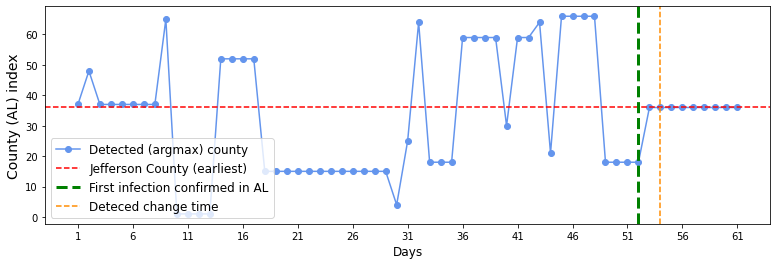}
    \caption{[Upper left] Noisy COVID-19 daily infection of 67 counties in AL. 
    [Upper right] Detecting the earliest onset time via multi-stream robust statistic $\bar\Phi_n$ with $\bar f_{\ta, n} = \text{Pois}(1)$ and $\bar g_{\ta, n, \nu} = \text{Pois}(2)$. The detected change time is 52 (2020/3/14). 
    [Lower] Identifying county with the earliest onset (Jefferson) in AL, induced by $\arg\max_{\theta \in \Theta}$.}
    \label{fig:AL infection}
\end{figure}

\begin{figure}[!t]
	\centering
	\includegraphics[scale=0.45]{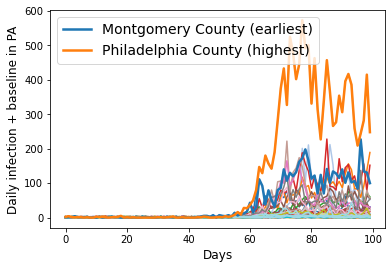}
	\includegraphics[scale=0.45]{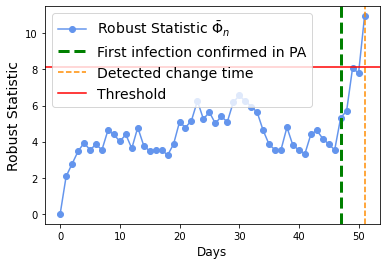}
	\includegraphics[scale=0.45]{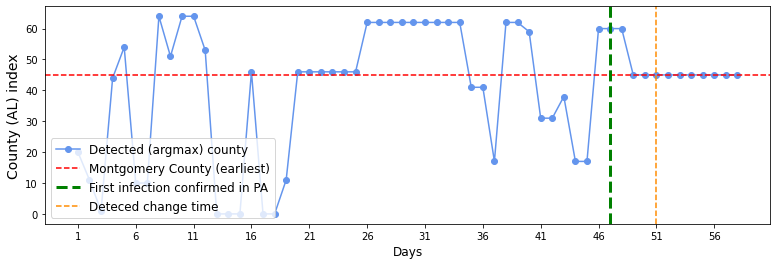}
    \caption{[Top left] Noisy COVID-19 daily infection of 67 counties in PA. 
    [Top right] Detecting the earliest onset time via multi-stream robust statistic $\bar\Phi_n$ with $\bar f_{\ta, n} = \text{Pois}(1)$ and $\bar g_{\ta, n, \nu} = \text{Pois}(2)$. The detected change time is 48 (2020/3/10). 
    [Lower] Identifying county with the earliest onset  (Montgomery) in PA, induced by $\arg\max_{\theta \in \Theta}$.}
    \label{fig:PA infection}
\end{figure}

We calculated $\bar\Phi_n$ in \eqref{eq:multstream_robust_onechange} with pre-LFL $\bar f_{\ta, n} = \text{Pois}(1)$ and post-LFL $\bar g_{\ta, n, \nu} = \text{Pois}(2)$. The threshold is $\log(|\Theta|/\alpha) = \log(50 \times 67)$. 
In Figure \ref{fig:AL infection} [Upper right] and Figure \ref{fig:PA infection} [Upper right], the multi-stream robust statistic detects the change quickly. In Figure \ref{fig:AL infection} [Lower] and Figure \ref{fig:PA infection} [Lower], counties with the earliest onsets are correctly identified by \eqref{eq:argmax_to_identify_stream}, with 7 additional observations taken after detected change time to ensure stable result.

\section{Conclusion}
\label{sec:conc}
We developed optimal robust algorithms for the quickest detection of changes in multi-stream non-stationary processes. We allowed for the possibility of unknown distributions both before and after the change in each stream. Our robust algorithms are based on the notion of the least favorable distribution or law (LFL). Our definition of LFL is different from the definitions used in classical QCD literature because we also allow our post-change statistical model to depend on the change point. Because of this dependence, we also need a stronger notion of monotone likelihood ratios to prove optimality. Thus, we provide a novel set 
of sufficient conditions for robust optimality in a multi-stream non-stationary setting. We also showed that these sufficient conditions are satisfied by Gaussian or Poisson count multi-stream data, and then used our algorithm for detecting changes in COVID and flight data.

\appendix

\section{Appendix}

\subsection{Proof of Theorem \ref{thm:LFLRobust_minimax}}\label{appen:1}
For the mean time to a false alarm, we show that the stopping rule defined in \eqref{eq:tau_t_star},  $\tau_{ss}$, satisfies the constraint in \eqref{eq:robustProbmini}. 
For the delay analysis, we show that the stopping rule $\tau_{ss}$ achieves the worst-case detection delay (over all possible pre- and post-change laws $F$ and $G$) when the pre-change distribution sequence is $\bar F = \{\bar f_n\}$, and the post-change distribution sequence is $\bar G = \{\bar g_{n, \nu}\}$. This ensures that the worst-case delay of $ \tau_{ss}$ is lower than the worst-case delay of any other stopping rule satisfying the false alarm constraint in \eqref{eq:robustProbmini}. 

\subsubsection{Mean Time to a False Alarm Analysis for Theorem \ref{thm:LFLRobust_minimax}}
To show that $\tau_{ss}$ satisfies the constraint in \eqref{eq:robustProbmini}, we show that the stopping rule has a mean time to a false alarm higher than $1/\alpha$, for all pre-change laws $F \in \mf$. For this purpose, we assume that the change does never occurs, i.e., $\nu = \infty$.

Applying Lemma \eqref{lem:stocbound_UV} with \eqref{eq:stocbounded-1}, and using the conditions of the theorem, we have
\begin{equation}\label{eq:cusum case 1}
    \begin{split}
        \Prob_{\infty}^{\bar F}(\tau_{ss} \leq N) 
        &= \Prob_{\infty}^{\bar F}[h(X_1, \dots, X_N) \geq 0]\\
        &\stackrel{(1a)}{\geq} \Prob_{\infty}^{F}[h(X_1, \dots, X_N) \geq 0]
        =\Prob_{\infty}^{F}(\tau_{ss} \leq N), \quad \forall F \in \mf,
    \end{split}
\end{equation}
where the function $h(z_1, z_2, \dots, z_{N})$ is given by
	\begin{equation}\label{eq:cusum function h}
		\begin{split}
			h(z_1, z_2, \dots, z_N) 
			&= \max_{1 \leq n \leq N} \left(\max_{1 \leq k \leq n} \left(\sum_{i=k}^n \log \frac{\bar{g}_{i,k}(z_{i})}{\bar f_i(z_i)}\right)- \bar A_{n, \alpha}\right).
		\end{split}
	\end{equation}
The inequality $h(X_1, X_2, \dots, X_{N}) \geq 0$ in \eqref{eq:cusum case 1} means that the CUSUM statistics defined in \eqref{eq:W_t_n} exceeds the sequence of thresholds $\{ \bar A_{n, \alpha}\}$ before time $N$.
Inequality $(1a)$ in \eqref{eq:cusum case 1} is implied by Lemma \ref{lem:stocbound_UV}, with independent $X_1, \dots, X_N$, $U_i = X_i \sim \bar f_{n}$, $V_i = X_i \sim f_{n}$ ($\bar f_{n} \succ {f}_{n}$ from \eqref{eq:stocbounded-1}),  and $h(X_1, X_2, \dots, X_{N})$ is a continuous and increasing function of these variables due to continuity and monotonicity of likelihood ratios and the sum and maximum operations.

Thus, we have the following:
\begin{align*}
    \Prob_{\infty}^{F}(\tau_{ss} \leq N) \leq \Prob_{\infty}^{\bar F}(\tau_{ss} \leq N), \quad \forall N \geq 0.
\end{align*}
After multiplying by $-1$, then adding $1$, summing over all $N \geq 0$, and invoking the assumption (b) in Theorem \eqref{thm:LFLRobust_minimax}, we have
\begin{align}
    \Expect_{\infty}^F(\tau_{ss}) \geq \Expect_{\infty}^{\bar F}(\tau_{ss}) = \frac{1}{\alpha}, \quad \forall F \in \mf.
\end{align}
Taking infimum over all $F \in \mathcal{F}$,  we get
\begin{align*}
    \inf_{F\in\mf} \Expect_{\infty}^F(\tau_{ss}) \geq \frac{1}{\alpha},
\end{align*}
and the mean time to a false alarm rate is higher than $1/\alpha$ for all $F \in \mf$. Thus, the constraint in \eqref{eq:robustProbmini} is satisfied by $\tau_{ss}$. 

\subsubsection{Detection delay analysis for Theorem \ref{thm:LFLRobust_minimax}}
We now show that the stopping rule $\tau_{ss}$ achieves the worst-case detection delay over all pre- and post-change laws when the laws are given by the LFLs. For this purpose, we consider $\nu=k$, $1 \leq k < \infty$.
The key step in the proof is to show that for each $k \in \mathbb{N}$,
\begin{equation}\label{eq:cusum keystep}
	\begin{split}
	\Expect_k^{\bar F, \bar{G}}&\left[(\tau_{ss} - k + 1)^+ | \mathscr{F}_{k-1}\right] 
   \geq \Expect_k^{F, G}\left[(\tau_{ss} - k+ 1)^+ | \mathscr{F}_{k-1}\right], \quad \forall (F, G) \in \mf\times\mg.
	\end{split}
\end{equation}
If the above statement is true, by the definition of essential supremum, we have the following for all $(F, G) \in \mf\times\mg$,
\begin{equation}\label{eq:esssuporder}
	\begin{split}
		\esssup \Expect_k^{\bar F, \bar{G}}&\left[(\tau_{ss} - k + 1)^+ | \mathscr{F}_{k-1}\right] 
        \geq \esssup \Expect_k^{F, G}\left[(\tau_{ss} - k+ 1)^+ | \mathscr{F}_{k-1}\right].
	\end{split}
\end{equation}
Since \eqref{eq:esssuporder} is true for every $k \geq 1$ and all $(F, G) \in \mf\times\mg$, 
\begin{equation*}
		\begin{split}
		\sup_k\esssup \Expect_k^{\bar F, \bar{G}}&\left[(\tau_{ss} - k + 1)^+ | \mathscr{F}_{k-1}\right] 
        \geq \sup_k \esssup \Expect_k^{F, G}\left[(\tau_{ss} - k+ 1)^+ | \mathscr{F}_{k-1}\right].
		\end{split}
\end{equation*}
This gives us
\begin{equation*}
		\sup_k\esssup\Expect_k^{\bar F, \bar{G}} \left[(\tau_{ss} - k + 1)^+ | \mathscr{F}_{k-1}\right] 
        = \sup_{(F, G) \in \mf\times\mg} \sup_k \esssup \Expect_k^{F, G}\left[(\tau_{ss} - k+ 1)^+ | \mathscr{F}_{k-1}\right].
\end{equation*}
By the Definition \eqref{eqn:wadd definition}, we have the core statement regarding worst-case delay,
\begin{equation}\label{eqn:WADD equality}
 \text{WADD}^{\bar F, \bar{G}}(\tau_{ss})=\sup_{(F, G)\in \mf\times\mg} \text{WADD}^{F, G}(\tau_{ss}).
\end{equation}
Thus, the supremum on the right is achieved at the LFLs. 

Now, if $\tau$ is any stopping rule satisfying the mean time to a false alarm constraint of $1/\alpha$ uniformly over $F$, i.e., satisfies \eqref{eq:robustProbmini}, then by the optimality or asymptotic optimality assumption on $\tau_{ss}$ for the LFLs, we have 
	\begin{equation}
		\begin{split}
			\sup_{(F, G)\in\mf\times\mg} \text{WADD}^{F, G}({\tau})
            \geq \text{WADD}^{\bar F, \bar{G}}({\tau})
            &\geq \text{WADD}^{\bar F, \bar{G}}(\tau_{ss}) (1+o^*(1))\\
			&= \sup_{(F, G)\in\mf\times\mg} \text{WADD}^{F, {G}}(\tau_{ss})(1+o^*(1)).
		\end{split}
	\end{equation}
Here, the term $o^*(1)$ is defined as 
    \begin{equation*}
        o^*(1) = 
        \begin{cases}
                0, &\quad \text{if $\tau_{ss}$ is exactly optimal,}\\
                o(1), &\quad \text{if $\tau_c$ is asymptotically optimal}.\\
        \end{cases}
    \end{equation*}
Thus, if we can prove the key step \eqref{eq:cusum keystep}, the above equation shows the robust optimality of the stopping rule $\tau_{ss}$.

We now prove the key step \eqref{eq:cusum keystep}. We prove this by showing that
\begin{equation}
\label{eq:cusum case 2}
\begin{split}
    \Expect_k^{\bar{F}, \bar{G}} \left[(\tau_{ss} - k+1)^+ \right | \mathscr{F}_{k-1}] 
    &\stackrel{(1b)}{=} \Expect_k^{F, \bar{G}}\left[(\tau_{ss} - k+1)^+ \right | \mathscr{F}_{k-1}]\\
    &\stackrel{(1c)}{\geq} \Expect_k^{F, G}\left[(\tau_{ss} - k+1)^+ \right | \mathscr{F}_{k-1}].
\end{split}
\end{equation}
The equality $(1b)$ in \eqref{eq:cusum case 2} is true because for CUSUM-type schemes, once we condition on the past realizations, the conditional expectation given no longer depends on the distribution of those realizations. 

Next, we prove the inequality $(1c)$ in \eqref{eq:cusum case 2}. To prove this, we show that for integers $N\geq 0$, $k \geq 1$, and all $(F, G) \in \mf\times\mg$,
	\begin{equation}\label{eq:keystep44}
		\begin{split}
			\Prob_k^{F, \bar{G}}\left[(\tau_{ss} - k+ 1)^+ > N |\mathscr{F}_{k-1}\right] 
            &\geq \Prob_k^{F, G}\left[(\tau_{ss} - k+ 1)^+ > N |\mathscr{F}_{k-1}\right].
		\end{split}
	\end{equation}
We prove this inequality for $N=0$ and $N \geq 1$ separately. 

For $N=0$, note that for arbitrary $\tau$,
\begin{align}\label{eqn:event equality 1}
    \{({\tau} - k + 1)^+ > 0\} = \{{\tau} - k + 1 > 0\} = \{{\tau} >  k- 1\} = \{{\tau} \leq  k- 1\}^c. 
\end{align}
Since the event $\{{\tau} \leq  k- 1\}$ is $\mathscr{F}_{k-1}$ measurable, both sides of the inequality \eqref{eq:keystep44} are equal to the indicator of the event $\{(\tau_{ss} - k + 1)^+ > 0\}$. This proves \eqref{eq:keystep44} for $N=0$. 
 
To prove \eqref{eq:keystep44} for $N \geq 1$, we show that for $N \geq 1$, $k \geq 1$, 
 \begin{equation*}
		\begin{split}
		\Prob_k^{F, \bar{G}}\left[(\tau_{ss} - k + 1)^+ \leq N  | \mathscr{F}_{k-1}\right] 
        &\leq \Prob_k^{F, G}\left[(\tau_{ss} - k+ 1)^+ \leq N | \mathscr{F}_{k-1} \right], \quad  \forall G \in \mathcal{G}.
		\end{split}
	\end{equation*}
First note that for arbitrary $\tau$, the events $\{(\tau - k+1)^+ \leq N\}$ and $\{(\tau - k+1) \leq N\}$ are identical when $N \geq 1$:
\begin{equation}
\label{eq:eventequiv_1}
    \begin{split}
       &\{(\tau - k+1)^+ \leq N\} \\
       =\ & \left[ \{(\tau - k+1)^+ \leq N\} \cap \{\tau \geq k-1 \}\right] \cup \left[ \{(\tau - k+1)^+ \leq N\} \cap \{\tau < k-1 \}\right]\\
       =\ & \left[ \{(\tau - k+1) \leq N\} \cap \{\tau \geq k-1 \}\right] \cup \left[ \{(\tau - k+1)^+ \leq N\} \cap \{\tau < k-1 \}\right]\\
       \stackrel{(1d)}{=}\ & \left[ \{(\tau - k+1) \leq N\} \cap \{\tau \geq k-1 \}\right] \cup \left[  \{\tau < k-1 \}\right]\\
       \stackrel{(1e)}{=}\ & \left[ \{(\tau - k+1) \leq N\} \cap \{\tau \geq k-1 \}\right] \cup \left[ \{(\tau - k+1) \leq N\} \cap \{\tau < k-1\}\right]\\
       =\ & \{(\tau - k+1) \leq N\}.
    \end{split}
\end{equation}
Here the equalities $(1d)$ and $(1e)$ follow because for $N \geq 0$,
\begin{equation*}
    \begin{split}
        \{\tau < k-1 \}  &\subset \{(\tau - k+1)^+ \leq N\}, \\
        \{\tau < k-1 \}  &\subset \{(\tau - k+1) \leq N\}.
    \end{split}
\end{equation*}
This equivalence of events implies that
	\begin{equation}\label{eq:temp1-1}
		\begin{split}
			\Prob_k^{F, \bar{G}}\left[(\tau_{ss} - k + 1)^+ \leq N |\mathscr{F}_{k-1}\right] 
			&= \Prob_k^{F, \bar{G}}\left[\tau_{ss} \leq k+N -1 |\mathscr{F}_{k-1}\right] \\
			& = \Prob_k^{F, \bar{G}}\left[h(X_1, X_2, \dots, X_{k+N-1}) \; \geq \; 0 \; | \; \mathscr{F}_{k-1}\right],
		\end{split}
	\end{equation}
where the function $h(z_1, \dots, z_{k+N-1})$ is defined as in \eqref{eq:cusum function h}: 
\begin{equation}\label{eq:cusum function h_1}
		\begin{split}
			h(z_1, z_2, \dots, z_{k+N-1}) 
			&= \max_{1 \leq n \leq k+N-1} \left(\max_{1 \leq k \leq n} \left(\sum_{i=k}^n \log \frac{\bar{g}_{i,k}(z_{i})}{\bar f_i(z_i)}\right)- \bar A_{n, \alpha}\right).
		\end{split}
	\end{equation}
The inequality $h(X_1, \dots, X_{k+N{-1}}) \geq 0$ in (\ref{eq:temp1-1}) means that the statistics defined in (\ref{eq:W_t_n}) exceeds the sequence of thresholds $\{ \bar A_{n, \alpha}\}$ before or at time $k + N{-1}$.
 
Using the stochastic boundedness assumption \eqref{eq:stocbounded-2}, 
 	\begin{equation*}
		\begin{split}
			g_{n, \nu}&\succ 	\bar{g}_{n, \nu}, \quad \forall g_{n, \nu} \in \mathcal{P}^G_{n, \nu}, \quad n, \nu=1,2, \dots,
		\end{split}
	\end{equation*}
we have
	\begin{equation}\label{eq:temp4}
		\begin{split}
			\Prob_k^{F,\bar{G}}\left[(\tau_{ss} - k+1)^+ \leq N \; | \; \mathscr{F}_{k-1}\right] 
			&= \Prob_k^{F,\bar{G}}\left[\tau_{ss} \leq k+N-1\; |\; \mathscr{F}_{k-1}\right] \\
			& = \Prob_k^{F,\bar{G}}\left[h(X_1, X_2, \dots, X_{k+N-1}) \geq 0 \; |\; \mathscr{F}_{k-1}\right] \\
			&\stackrel{(1f)}{\leq} \Prob_k^{F,{G}}\left[h(X_1, X_2, \dots, X_{k+N-1}) \geq 0 \; |\; \mathscr{F}_{k-1}\right] \\ 
            &= \Prob_k^{F,G}\left[\tau_{ss} \leq k+N-1 \; |\; \mathscr{F}_{k-1}\right] \\
			&=\Prob_k^{F,{G}}\left[(\tau_{ss} - k+1)^+ \leq N \; |\; \mathscr{F}_{k-1}\right].
		\end{split}
	\end{equation}
The inequality $(1f)$ is implied by Lemma \ref{lem:stocbound_UV}. 
The first $k-1$ variables of $(X_1, \dots, X_{k+N})$ are fixed after conditioning on $\mathscr{F}_{k-1}$ and the remaining $X_k, \dots, X_{k+N-1}$ are independent, with their law governed by $\bar{G}$ to the left of inequality, and by $G$ to the right ($g_{n, \nu} \succ \bar g_{n, \nu}$ from \eqref{eq:stocbounded-2}).
Besides, the function $h(X_1, X_2, \dots, X_{k+N-1})$ defined in \eqref{eq:cusum function h_1} is continuous and increasing in these independent variables $(X_1, \dots, X_{k+N})$, since finite number of maximum and sum operations preserve continuity and the monotone and continuous likelihood ratio assumptions stated in the theorem statement.
Therefore, we can invoke Lemma \ref{lem:stocbound_UV}. This proves \eqref{eq:keystep44}, and hence \eqref{eq:cusum case 2} and \eqref{eq:cusum keystep}.

\subsection{Proof of Theorem \ref{thm:2}}\label{appen:2}
The proof is based on similar arguments as the proof of Theorem \ref{thm:LFLRobust_minimax} given above. However, the notations are more involved. We provide the proof here for completeness. 

\subsubsection{Mean time to a false alarm analysis for Theorem \ref{thm:2}}
We prove that the stopping rule $\tau_{ms}$ in \eqref{eqn: tau glr B} has a mean-time to false alarm higher than $1/\alpha$, for all $F \in \mf$. 
We have
\begin{equation}\label{eq:cusum case 11}
    \begin{split}
        \Prob_{\infty}^{\bar{F}}(\tau_{ms} \leq N) 
        &= \Prob_{\infty}^{\bar{F}}[l(\{X_{\ta, 1}, \dots, X_{\ta,N}\}_{\ta\in\Theta}) \geq 0]\\
        &\stackrel{(2a)}{\geq} \Prob_{\infty}^{F}[l(\{X_{\ta,1}, \dots, X_{\ta,N}\}_{\ta\in\Theta}) \geq 0]
        =\Prob_{\infty}^{F}(\tau_{ms} \leq N), \quad \forall F \in \mf,
    \end{split}
\end{equation}
where the function $l(\{z_{\ta,1}, \dots, z_{\ta,N}\}_{\ta\in\Theta})$ is given by
	\begin{equation*}
		\begin{split}
			l(\{z_{\ta,1}, \dots, z_{\ta,N}\}_{\ta\in\Theta}) 
		=  \max_{1 \leq n \leq N} \left(\max_{1 \leq k \leq n}\max_{B \in \mb} \left(\sum_{\ta\in B}\sum_{i=k}^n \log\frac{\bar{g}_{\theta, i,k}(z_{\ta, i})}{\bar{f}_{\theta,i}(z_{\ta, i})}\right)-{{A}_{n, \alpha}}\right).
		\end{split}
	\end{equation*}
Thus, $l(\{X_{\ta, 1}, \dots, X_{\ta,N}\}_{\ta\in\Theta}) \geq 0$ in \eqref{eq:cusum case 11} means that the multi-stream robust test statistic \eqref{G_n_B} exceeds the sequence of thresholds $\{A_{n, \alpha}\}$ before time $N$.
Inequality $(2a)$ in \eqref{eq:cusum case 11} follows from Lemma \eqref{lem:stocbound_UV}, the fact that $(\bar{F}, \bar{G})$ are LFLs, and the monotone and continuous likelihood ratio assumptions stated in the theorem statement.
Thus, we have the following:
\begin{align*}
    \Prob_{\infty}^{F}(\tau_{ms} \leq N) \leq \Prob_{\infty}^{\bar{F}}(\tau_{ms} \leq N), \quad \forall N \geq 1.
\end{align*}
After multiplying by $-1$, then adding $1$, summing over all $N \geq 0$, and invoking the assumption \eqref{eqn:E_infty_glr multi 2}, we get 
\begin{align}
\nonumber
    \Expect_{\infty}^{F}(\tau_{ms}) \geq \Expect_{\infty}^{\bar{F}}(\tau_{ms}) = \frac{1}{\alpha}, \quad \forall F \in \mf.
\end{align}
Since $\bar{F} \in \mf$, we have the mean time to a false alarm higher than $1/\alpha$ for all $F \in \mf$,
\begin{align*}
    \inf_{F\in\mf} \Expect_{\infty}^{F}(\tau_{ms}) \geq \frac{1}{\alpha}.
\end{align*}

\subsubsection{Detection delay analysis for Theorem \ref{thm:2}}
We show that the stopping rule $\tau_{ms}$ in \eqref{eqn: tau glr B} achieves the worst-case detection delay when the distribution is LFL. 
Without loss of generality, suppose $B$ is the true but fixed post-change parameter set, where the changes occur in stream $\theta\in B$. 

For any stopping rule $\tau$ satisfying $\inf_{F\in\mf} \Expect_{\infty}^{F}(\tau) \geq \frac{1}{\alpha}$, since 
$\tau_{ms}$ is optimal for problem \eqref{eq:QCDproblem-B} with $\mf = \{\bar{F}\}$ and $\mg = \{\bar{G}\}$, we have 
\begin{align} \label{eq: optimal performance}
    \text{WADD}^{\bar F, \bar G, B}(\tau)
    \geq
    \text{WADD}^{\bar{F}, \bar{G}, B}(\tau_{ms})(1 + o^*(1)).
\end{align}
The term $o^*(1)$ is defined as
        \begin{equation*}
            o^*(1) = 
            \begin{cases}
                0, &\quad \text{if $\tau_{ms}$ is exactly optimal,}\\
                o(1), &\quad \text{if $\tau_{ms}$ is asymptotically optimal}.\\
            \end{cases}
        \end{equation*}
Hence, we have
\begin{align} \label{eqn:mid result 1}
\begin{split}
    \sup_{F \in\mf, G \in\mg } \text{WADD}^{F, G, B}(\tau)
    &\geq \text{WADD}^{\bar F, \bar G, B}(\tau)\\
    &\geq \text{WADD}^{\bar{F}, \bar{G}, B}(\tau_{ms})(1 + o^*(1)).
\end{split}
\end{align}
To complete the proof, we will show that
\begin{align}\label{eqn:mid result 2}
    \text{WADD}^{\bar{F}, \bar{G}, B }(\tau_{ms}) 
    \geq \sup_{F \in\mf, G \in\mg }
    \text{WADD}^{F, G, B}(\tau_{ms}).
\end{align}
If this is true, we can substitute it in \eqref{eqn:mid result 1} to get
\begin{align*}
    \sup_{F \in\mf, G \in\mg } \text{WADD}^{F, G, B}(\tau)
    \geq
    \sup_{F \in\mf, G \in\mg }
    \text{WADD}^{F, G, B}(\tau_{ms})(1 + o^*(1)).
\end{align*}
Thus, $\tau_{ms}$ will be robust optimal to problem \eqref{eqn:robustprob2}. 

We are left with proving the core statement \eqref{eqn:mid result 2}. Similar to the single-stream case, the key step is the following: for each $k \in \mathbb{N}$ and all $(F, G) \in \mf\times\mg $,
\begin{equation}\label{eq:keystep31}
\begin{split}
\Expect_{k}^{\bar{F}, \bar{G}, B }\left[(\tau_{ms} - k+1)^+ | \mathscr{F}_{k-1}\right] 
   &\stackrel{(2b)}{=}
   \Expect_{k}^{F, \bar{G}, B }\left[(\tau_{ms} - k+1)^+ | \mathscr{F}_{k-1}\right]\\
   &\stackrel{(2c)}{\geq} 
   \Expect_{k}^{F, {G}, B }\left[(\tau_{ms} - k+1)^+ | \mathscr{F}_{k-1}\right].
\end{split}
\end{equation}
The justification for (2b) is similar to that given in the single-stream case. Thus, we only need to establish the the inequality (2c). This is because if we have (2c), then we have the following series of inequalities: 
\begin{equation}
    \begin{split}
        \Expect_{k}^{\bar{F}, \bar{G}, B }\left[(\tau_{ms} - k+1)^+ | \mathscr{F}_{k-1}\right] 
   &\geq
   \Expect_{k}^{F, {G}, B }\left[(\tau_{ms} - k+1)^+ | \mathscr{F}_{k-1}\right], \; \forall F, {G} \\
   \text{WADD}^{\bar{F}, \bar{G}, B }(\tau_{ms}) 
    &\geq 
    \text{WADD}^{F, G, B}(\tau_{ms}), \; \forall F, {G}\\
     \text{WADD}^{\bar{F}, \bar{G}, B }(\tau_{ms}) 
    &\geq \sup_{F \in\mf, G \in\mg }
    \text{WADD}^{F, G, B}(\tau_{ms}).
    \end{split}
\end{equation}
This gives us the core statement \eqref{eqn:mid result 2} because $\bar F \in\mf$, and  $\bar G \in\mg $. 

Thus, the proof is completed if we prove (2c) in \eqref{eq:keystep31}. 
For it, we need to show that for every integer $N\geq 0$, and all $G  \in \mathcal{G} $ and $k \geq 1$, 
	\begin{equation}
 \label{eq:keystep4}
		\begin{split}
            \Prob_{k}^{F, \bar{G}, B }\left[(\tau_{ms} - k+1)^+ > N \mid \mathscr{F}_{k-1}\right] \geq \Prob_{k}^{F, G, B}\left[(\tau_{ms} - k+1)^+ > N \mid \mathscr{F}_{k-1}\right].
		\end{split}
	\end{equation}
We prove this inequality for $N=0$ and $N \geq 1$ separately. 
For $N=0$, by \eqref{eqn:event equality 1} and event $\{\tau_{ms} \leq k - 1\}$ is $\mathscr{F}_{k-1}$-measurable, both sides of the inequality \eqref{eq:keystep4} are equal to the indicator of the event $\{(\tau_{ms} - k + 1)^+ > N\}$. 

For $N \geq 1$, we show the equivalent statement of \eqref{eq:keystep4} for $k \geq 1$ and all $G  \in \mathcal{G} $,
 \begin{equation}\label{eq:keystep33}
		\begin{split}
		\Prob_k^{F, \bar{G}, B }\left[(\tau_{ms} - k + 1)^+ \leq N  | \mathscr{F}_{ k-1}\right] 
        &\leq 
        \Prob_k^{F, G, B}\left[(\tau_{ms} - k+ 1)^+ \leq N | \mathscr{F}_{ k-1} \right].
		\end{split}
	\end{equation}
By the equivalence in \eqref{eq:eventequiv_1}, we have
	\begin{align*}
		\Prob_{k}^{F, \bar{G}, B }\left[(\tau_{ms} - k+1)^+ \leq N \mid \mathscr{F}_{k-1}\right]
		&= \Prob_{k}^{F, \bar{G}, B }\left[\tau_{ms} \leq k+N-1 \mid \mathscr{F}_{k-1}\right] \\
		&=\Prob_{k}^{F, \bar{G}, B }\left[l(\{X_{\ta, 1}, \dots, X_{\ta,k+N-1}\}_{\ta\in\Theta}) \; \geq \; 0 \mid \mathscr{F}_{k-1}\right],
	\end{align*}
where the function $l(\{z_{\ta, 1}, \dots, z_{\ta,k+N-1}\}_{\ta\in\Theta})$ is given by
	\begin{equation*}
			l(\{z_{\ta, 1}, \dots, z_{\ta,k+N-1}\}_{\ta\in\Theta}) 
		=  \max_{1 \leq n \leq N} \left(\max_{1 \leq k \leq n}\max_{B \in \mb} \left(\sum_{\ta\in B}\sum_{i=k}^n \log\frac{\bar{g}_{\theta, i,k}(z_{\ta, i})}{\bar{f}_{\theta,i}(z_{\ta, i})}\right)-{{A}_{n, \alpha}}\right).
	\end{equation*}

This gives us
\begin{equation}
\begin{split}
   \Prob_{k}^{F, \bar{G}, B }&\left[(\tau_{ms} - k+1)^+ \leq N \mid \mathscr{F}_{ k-1}\right]\\
=\ &\Prob_{k}^{F, \bar{G}, B }\left[\tau_{ms} \leq k+N-1 \mid \mathscr{F}_{ k-1}\right] \\
=\ &\Prob_{k}^{F, \bar{G}, B }\left[l(\{X_{\ta, 1}, \dots, X_{\ta,k+N-1}\}_{\ta\in\Theta}) \geq 0 \mid \mathscr{F}_{ k-1}\right] \\
\stackrel{(2d)}{\leq}\ &\Prob_{k}^{F, {G}, B }\left[l(\{X_{\ta, 1}, \dots, X_{\ta,k+N-1}\}_{\ta\in\Theta}) \geq 0 \mid \mathscr{F}_{ k-1}\right] \\ 
=\ &\Prob_{k}^{F, {G}, B }\left[\tau_{ms} \leq k+N-1 \mid \mathscr{F}_{ k-1}\right] \\
=\ &\Prob_{k}^{F, {G}, B }\left[(\tau_{ms} - k+1)^+ \leq N \mid \mathscr{F}_{ k-1}\right].
\end{split}		
\end{equation}
Inequality $(2d)$ is true, because the conditioning on the change point and the past realization $\mathscr{F}_{k-1}$ fixes the first $k-1$ coordinates of $(\{X_{\ta, 1}, \dots, X_{\ta,k+N-1}\}_{\ta\in\Theta})$ for each $\theta$. For $\ta\notin B$, the law remains unchanged. For $\ta\in B$ where changes occur, the rest of the coordinates are independent, with their law being $\bar{G} $ to the left of inequality, and $G $ to the right. 
Inequality $(2d)$ then follows from Lemma \eqref{lem:stocbound_UV}, the fact that $(\bar{F}, \bar{G})$ are LFLs, the fact that finite number of maximum and sum operations preserve continuity, and the monotone and continuous likelihood ratio assumptions stated in the theorem statement.

This proves \eqref{eq:keystep33}, \eqref{eq:keystep4} and hence \eqref{eq:keystep31}.

\section*{Acknowledgment}
Taposh Banerjee was supported in part by the U.S. National Science Foundation under grant 2427909. Hoda Bidkhori was supported in part by the U.S. National Science Foundation under grant 2427910. 

\bibliographystyle{IEEEtran}
\bibliography{QCD.bib}

\end{document}